\documentclass{eptcs}
\usepackage{xcolor}
\usepackage{mathtools,amssymb,amsmath,amsfonts,amsxtra,amsthm,mathrsfs,mathpartir}
\usepackage{array}
\usepackage{cmll}
\usepackage[toc,page]{appendix}
\usepackage{graphicx}
\usepackage{varioref}
\usepackage{verbatim,comment}
\usepackage{longtable}
\usepackage{tikz}

\usetikzlibrary{chains,fit,shapes}

\newcolumntype{L}{>{$}l<{$}}
\newcolumntype{R}{>{$}r<{$}}
\newcolumntype{C}{>{$}c<{$}}
\newcommand{\mytikshrink}{0.8}

\theoremstyle{plain}
\newtheorem{theorem}{Theorem}
\newtheorem{lemma}[theorem]{Lemma}
\newtheorem{corollary}[theorem]{Corollary}

\theoremstyle{remark}
\newtheorem{remark}[theorem]{Remark}

\theoremstyle{definition}
\newtheorem{notation}[theorem]{Notation}
\newtheorem{definition}[theorem]{Definition}

\newenvironment{tabred}{
  \begin{center}
  \begin{tabular}{L L L}
}{
  \end{tabular}
  \end{center}
}

\newcommand{\vect}[1]{\overrightarrow{#1}}

 \newcommand{\sur}[2]{\overset{#2}{#1}}
 
 \newcommand{\entre}[3]{#2 \leq #1
  \leq #3} \newcommand{\ptentre}[3]{(\forall #1\ : \entre{#1}{#2}{#3})}

\newcommand{\metas}[2]{#2\{{#1}\}}
\newcommand{\varmetas}[3]{#3\{{#2}/{#1}\}}
\newcommand{\varsmetas}[3]{\varmetas{\vect{#1}}{\vect{#2}}{#3}}

\newcommand{\varss}[2]{{#2}[{#1}]}
\newcommand{\vars}[3]{{#3}[{#2} / {#1}]}
\newcommand{\sind}[3]{{#2}[{#1}]_{#3}}

\newcommand{\refrawsd}[3]{\sind{#2}{#3}{\downarrow}}
\newcommand{\refsrawsd}[3]{\sind{#2}{#3}{\downarrow}}
\newcommand{\refssd}[3]{\refsrawsd{#1}{\multv{#2}}{#3}}

\newcommand{\refrawsu}[3]{\sind{#2}{#3}{\uparrow}}
\newcommand{\refsrawsu}[3]{\sind{#2}{#3}{\uparrow}}
\newcommand{\refssu}[3]{\sind{\multv{#2}}{#3}{\uparrow}}

\newcommand{\refrawsl}[3]{\sind{#2}{#3}{\lambda}}
\newcommand{\refsrawsl}[3]{\sind{#2}{#3}{\lambda}}
\newcommand{\refssl}[3]{\sind{\multv{#2}}{#3}{\lambda}}

\newcommand{\multv}[1]{\mathcal{#1}}
\newcommand{\get}[1]{\text{get}({#1})}
\newcommand{\set}[2]{\text{set}({#1},{#2})}
\newcommand{\store}[2]{{#1} \Leftarrow {#2}}

\newcommand{\rrulet}[2]{($\texttt{#1}_\texttt{#2}$)}
\newcommand{\rrule}[2]{($\texttt{#1}_{#2}$)}

\newcommand{\reach}[2]{\text{Reach}({#1},{#2})}
\newcommand{\skel}[1]{\text{Sk}({#1})}
\newcommand{\rincl}{\sqsubseteq}
\newcommand{\rdiffraw}[1]{\rincl_{#1}}
\newcommand{\rdiff}[1]{\rdiffraw{\multv{#1}}}
\newcommand{\req}{\backsimeq}

\newcommand{\seq}[1]{({#1}_i)}
\newcommand{\dom}[1]{\text{dom}({#1})}

\newcommand{\lamadio}{$\lambda_\text{C}$}
\newcommand{\lthis}{$\lambda_\text{cES}$}

\newcommand{\setw}{assignment}
\newcommand{\getw}{read}

\newcommand{\condSN}{(\textbf{SN})}
\newcommand{\condWB}{(\textbf{WB})}

\newcommand{\mynl}{\\[1.5ex]}

\title{An Interactive Proof of Termination for a Concurrent $\lambda$-calculus
  with References and Explicit Substitutions}
\author{Yann Hamdaoui
\institute{IRIF, Univ. Paris Diderot}
\email{yann.hamdaoui@irif.fr}
\and
Beno{\^i}t Valiron
\institute{LRI – CentraleSupelec, Univ. Paris Saclay}
\email{benoit.valiron@lri.fr}
}
\def\titlerunning{An Interactive Proof of Termination}
\def\authorrunning{Y. Hamdaoui \& B. Valiron}

\begin{document}

\maketitle


\begin{abstract}
  In this paper we introduce a typed, concurrent $\lambda$-calculus
  with references featuring explicit substitutions for variables and
  references. Alongside usual safety properties, we recover strong
  normalization. The proof is based on a reducibility technique and an
  original interactive property reminiscent of the Game
  Semantics approach.
\end{abstract}

\section{Introduction}
\label{sec:intro}

%

The $\lambda$-calculus is a versatile framework in the study and design of
higher-order functional programming languages. One of the reasons of
its widespread usage is the fact that it can easily be extended to
model various computational side-effects. Another reason comes
from its theoretical ground and the fine granularity it allows in the
design of abstract machines to express various reduction
strategies. These abstract-machines can then serve as foundation for
the design of efficient interpreters and compilers.

A specially useful tool in the design of such abstract machines is the
notion of explicit substitution, a refinement over $\beta$-reduction.
The $\beta$-reduction of the $\lambda$-calculus is a meta-rule where
substitution is defined inductively and performed all at once on the
term. But its implementation is a whole different story: to avoid
size explosion in presence of duplication,
mechanisms such as sharing are usually deployed. Abstract machines
implement various specific strategies that may either be representable
in pure $\lambda$-calculus (call-by-value or call-by-name) or for
which the syntax needs to be augmented with new objects
(e.g. call-by-need or linear head reduction). The mismatch between
$\beta$-reduction and actual
implementations can make the proof of soundness for an evaluator or a
compiler a highly nontrivial task. The heart of the theory of explicit
substitutions, introduced in~\cite{Abadi:1989:ES:96709.96712}, is to
give substitutions a first class status as objects of the syntax to
better understand the dynamics and implementation of
$\beta$-reduction. It consists in decomposing a substitution into
explicit atomic steps. The main ingredient is to modify the $\beta$
rule so that $(\lambda x.M)N$ reduces to $M[N/x]$, where $[N/x]$ is now part
of the syntax. Additional reduction rules are then provided to
propagate the substitution $[N/x]$ to atoms.

Studied for the last thirty
years~\cite{%
Abadi:1989:ES:96709.96712,
Accattoli:2015:PNC:2852784.2853038,
Accattoli:2014:NST:2535838.2535886,
Accattoli2016,
Accattoli2010,
GP99,
kesner2009theory,
LRD95,
Ros92,
SFM03
},
explicit substitution turns out to be a crucial device when
transitioning from a formal higher-order calculus to a concrete
implementation. It has been considered in the context of sharing of
mutual recursive definitions~\cite{Ros92},
higher-order unification~\cite{LRD95},
algebraic
data-types~\cite{GP99}, efficient abstract
machines~\cite{Accattoli:2015:PNC:2852784.2853038,SFM03},
cost-model analysis~\cite{Accattoli2016},
{\em etc}.
The use of explicit substitutions however comes at a
price~\cite{kesner2009theory}.  Calculi with such a feature are
sensitive to the definition of reduction rules. If one is too liberal
in how substitutions can be composed then a strongly normalizing
$\lambda$-term may diverge in a calculus with explicit
substitutions~\cite{Mellies1995}. If one is too restrictive, confluence
on metaterms is lost~\cite{Curien:1996:CPW:226643.226675}. The
challenge is to carefully design the language to implement desirable
features without losing fundamental properties. Several solutions have
been proposed to fix these
defects~\cite{Accattoli2010,kesner2009theory} for explicit
substitutions of term variables.

This paper introduces an extension of explicit substitutions to a
novel case: a lambda-calculus augmented with concurrency and
references. Such a calculus forms a natural model for shared memory
and message passing. We aim at proving that a translation of a shared
memory model to a message passing one is sound. The long time goal of
this work is to implement a memoryful language in a formalism (such
that interaction nets/proof nets) that allows to easily distribute
parts of a program to be executed on different nodes, and to
parallelize independent parts of the program.

The current paper concentrates on the problem of strong-normalization
of such a language.

\subsection{Strong Normalization in a Concurrent Calculus with
  References}\label{sec:amadio}

A concurrent lambda-calculus with references -- referred as {\lamadio}
below -- has been introduced by Amadio in~\cite{Amadio2009}. It is a
call-by-value $\lambda$-calculus extended with:
\begin{itemize}
\item a notion of threads and an operator $\parallel$ for parallel
  composition of threads,
\item two terms $\set{r}{V}$ and $\get{r}$, to respectively
  assign a value to and read from a reference,
\item special threads $\store{r}{V}$, called stores, accounting for
  assignments.
\end{itemize}
When $\set{r}{V}$ is reduced, it turns to the unit value $\ast$ and
produces a store $\store{r}{V}$ making the value available to all the
other threads. A corresponding construct $\get{r}$ is reduced by
choosing non deterministically a value among all the available
stores. For example, assuming some support for basic arithmetic
consider the program
$
(\lambda x. x + 1)\ \get{r} \parallel
  \set{r}{0} \parallel \set{r}{1}.
$
It consists of 3 threads: two concurrent assignments $\set{r}{0}$
and $\set{r}{1}$, and an application $(\lambda x. x + 1)\,\get{r}$.
This programs admits two normal forms depending on which
assignment ``wins'':  the term $1 \parallel
\ast \parallel \ast \parallel r \Leftarrow 0 \parallel
\store{r}{1}$ and the term $2 \parallel
\ast \parallel \ast \parallel r \Leftarrow 0 \parallel
\store{r}{1}$. Despite the $\parallel$ operator being a static constructor, it can be
embedded in abstractions and thus dynamically liberated or duplicated. For example, the term
$(\lambda f. f\ \ast \parallel f\ \ast)$ act like a \emph{fork} operation: if
applied to $M$, it generates two copy of its argument in two parallel threads $M\ \ast \parallel M\ \ast$.
Coupled with the Landin's fixpoint introduced below one can even write a fork
bomb, that is a non terminating term which spans an unbounded number of
threads.

In this language, the stores are global and cumulative: their scope is the
whole program, and each {\setw} adds a new binding that does not erase the
previous one. Reading from a store is a non deterministic process that
chooses a value among the available ones. References are able to handle an unlimited
number of values and are understood as a typed abstraction of possibly
several concrete memory cells. This feature allows {\lamadio} to simulate various
other calculi with references such as variants with dynamic references or
communication~\cite{madet:tel-00794977}.

While a simple type system for usual $\lambda$-calculus ensures
termination, the situation is quite different in a language with
higher-order references. The so called Landin's trick~\cite{landin1964mechanical} allows to encode
a fixpoint in the simply typed version of a calculus with references.
The problem lies in the fact that one can store in a reference $r$
values that can themselves read from the reference~$r$, leading
to a circularity. For example, the term $\get{r}\ \ast \parallel r
\Leftarrow (\lambda x. \get{r}\ \ast)$ loops while involving only simple types
$\mbox{Unit}$ and $\mbox{Unit} \to \mbox{Unit}$.

In order to address this issue, type and effects systems have been
introduced to track the potential effects produced by a term during
its evaluation. Together with stratification on
references~\cite{BOUDOL2010716}, one can recast termination in such an
imperative context. Intuitively, stratification imposes an order
between references: a reference can only store terms that
access smaller ones, ruling out Landin's fixpoint. Formally, this
allows to apply the usual reducibility argument to a calculus with
references: stratification ensures that the inductive definition of
reducibility sets on types with effects is well-founded.

While scheduling is explicitly handled through language constructs
in~\cite{BOUDOL2010716}, {\lamadio}'s liberal reduction allows to chose a
different thread to operate on at any time. This cause additional difficulty, as
from a single thread's point of view, arbitrary new {\setw s} may become
available between two reduction steps. For {\lamadio}, the proof of
termination in~\cite{Amadio2009} resorts to what amounts to infinite terms with
the notion of saturated stores.

\subsection{Our Contributions}

The contributions of this paper are twofold.
\begin{enumerate}
\item The definition of a system of \emph{explicit substitutions} for a
  concurrent $\lambda$-calculus with references, both for variables
  and references.

  \smallskip
  The problem we address is the bidirectional property of assignment
  of references within a term. An assignment for a term variable in a
  redex only diffuses inward: in $(\lambda x.M)V$, the assignment
  $x \mapsto V$ only concerns the subterm $M$. Instead, a reference
  assignment $\set{r}{V}$ is potentially {\em global}: it concerns
  all the occurrences of the subterm $\get{r}$.

  \smallskip
  Our first contribution is to propose an explicit substitution
  mechanism to be able to express reference assignment step-wise, as
  for term-variables.

\item A proof of strong normalization for a typed fragment using
  a novel \emph{interactive} property.

  \smallskip
  Akin to~\cite{Amadio2009}, the language we propose is typed and the
  type-system is enforcing strong-normalization. In the proof
  of~\cite{Amadio2009} the infinitary structure of terms is restricted to
  top-level stores. In our setting, this would require infinite
  explicit substitutions which are subject to duplication, erasure, composition,
  \ldots

  \smallskip
  Our proof only uses {\em finite} terms. It has a Game Semantics flavor which
  we find of interest on its own. Indeed, we use the idea of abstracting the
  context in which a subterm is executed as an opponent able to interact by
  sending and receiving explicit substitutions.  Moreover, we believe that the
  finite, interactive technique we develop in this second contribution may be
  well-suited for different settings such as proof nets or other concurrent
  calculi.
\end{enumerate}

\subsection{Plan of the paper}
Section~2 presents the calculus with explicit
substitutions {\lthis}. Section~3 introduces the stratified type and
effect system. Section~4 focuses on the proof of strong
normalization, while Section~5 discusses the construction. Section~6
concludes the paper.

\section{A Concurrent $\lambda$-calculus with Explicit
Substitutions}\label{section:lthis}

In standard presentations of the lambda-calculus and its extensions such
as~\cite{Amadio2009}, substitutions are applied globally. This hides
the implementation details of the procedure. Exposing such an
implementation is one of the reasons for the introduction of explicit
substitutions. In the literature, explicit substitutions have only been
used for term variables and not for references.

In this section, we introduce the language {\lthis}, a call-by-value,
concurrent $\lambda$-calculus with explicit substitutions for both
term variables and references.

\subsection{Syntax}

The language {\lthis} has two kinds of variables: term variables (simply named
{\em variables}) represented with $x,y,\ldots$, and {\em references},
represented with $r,r',\ldots$.  Substitutions are represented by partial
functions with finite support. Variable substitutions, denoted with Greek
letters $\sigma,\tau,\ldots$, map variables to values. Reference substitutions,
denoted with calligraphic uppercase letters $\multv{V},\multv{U},\ldots$,
map references to {\em finite multisets} of values. Multisets reflect the
non-determinism, as multiple writes may have been performed on the same
reference. They are represented with the symbol $\mathcal{E}$.
The language consists
of values, terms and sums of terms, representing non-determinism.
\begin{center}
\begin{tabular}{llll}
  -values & $V$ & $::=$ & $x \mid \ast \mid \lambda x.M$ \\
  -terms & $M$ & $::=$ & $V \mid \varss{\sigma}{M} \mid
                         \refssl{r}{V}{(M\,M)} \mid \get{r}
\mid \refssd{r}{V}{M} \mid \refssu{r}{V}{M} \mid M \parallel M$ \\
  -sums & $\mathbf{M}$ & $::=$ & $\mathbf{0} \mid M \mid \mathbf{M} +
  \mathbf{M}$ \\
\end{tabular}
\end{center}
The construct $\varss{\sigma}{M}$ stands for the explicit substitutions
of variables in $M$ under the substitution $\sigma$. There are three
constructs for explicit substitutions for references:
$\refssd{r}{V}{M}$ and $\refssu{r}{V}{M}$ are respectively the
downward and upward references substitutions, while
$\refssl{r}{V}{(M\,M)}$ is the $\lambda$-substitution. The reason for
which the language needs three distinct notations is explained in the
next section while presenting the reduction rules. Finally, the role
of the sum-terms $\mathbf{M}$ is to capture and keep all
non-deterministic behaviours.

Terms are considered modulo an equivalence relation presented in
Table~\ref{fig:lth-struct-rules}. The sum is
idempotent, associative and commutative, while the parallel composition
is associative and commutative.

\begin{remark}
  The three constructs $\refssd{r}{V}{}$, $\refssu{r}{V}{}$ and
  $\refssl{r}{V}{}$ encapsulate the assignments in terms and in stores
  presented in Section~\ref{sec:intro}: there is no need anymore for
  ${\rm set}(r,V)$ and $r\Leftarrow V$. See
  Section~\ref{sec:discussion} for a discussion of this aspect.
\end{remark}

\begin{notation}
  Reference substitutions will be sometimes written
  with the notation $[\store r V]$ to mean $[\multv{V}]$ with $\multv{V} : r
  \mapsto [V]$. Explicit variables substitutions are written $\varss{\{x_1\mapsto
    V_1,\ldots,x_n\mapsto V_n\}}{M}$.
  Finally, by abuse of notation we write $(M\ N)$ for $(M\ N)[]$.
\end{notation}

\subsection{Reduction}

We adopt a weak call-by-value reduction where the reduction order of
an application is not specified. It is weak in the sense that no
reduction occurs under abstractions.

Although in a general setting non-determinism and call-by-value taken together
may break confluence even when collecting all possible
outcomes~\cite{deLiguoroP95},
this phenomenon does not happen here. Indeed, we cannot reduce under
abstractions, and the only non-deterministic construct $\get{r}$ must be reduced
before being duplicated, avoiding problematic interactions between
$\beta$-reduction and non-deterministic choice.

The language {{\lthis}} is equipped with the reduction defined
in Table~\ref{fig:lth-red-rules}.
The rules presented are closed under the structural rules
  of Table~\ref{fig:lth-struct-rules}. We assume the usual conventions
  on alpha-equivalence of term, and as customary substitutions are considered
  modulo this alpha-equivalence.
They make use of several notations
that we lay out below.
Rules devoted to dispatching substitutions are referred as
\emph{structural rules}. The \emph{variable} (resp. \emph{downward},
\emph{upward}) \emph{structural rules} consist in \rrulet{subst}{}
(resp.  \rrulet{subst-r}{}, \rrulet{subst-r'}{}) rules excluding
\rrulet{subst}{var} (resp.  \rrulet{subst-r}{get},
\rrule{subst-r}{\top}). An in-depth discussion about these rules follows.

\begin{table}[tb]
\centering
\begin{minipage}{.49\linewidth}
\centering
\begin{tabular}{rcl}
$M \parallel M'$ & = & $M' \parallel M$ \\
$(M \parallel M') \parallel M''$ & = & $M \parallel (M' \parallel M'')$ \\
$\mathbf{M} + \mathbf{M}'$ & = & $\mathbf{M}' + \mathbf{M}$ \\
$(\mathbf{M} + \mathbf{M}') + \mathbf{M}''$ & = & $\mathbf{M} + (\mathbf{M}' +
\mathbf{M}'')$ \\
$\mathbf{M} + \mathbf{M}$ = $\mathbf{M} + \mathbf{0}$ & = & $\mathbf{M}$
\mynl
\end{tabular}
\caption{Structural Rules}\label{fig:lth-struct-rules}
\end{minipage}
\begin{minipage}{.49\linewidth}
\centering
\begin{tabular}{lcl}
~\\
$E$ & $::=$ & $[.] \mid \refssl{r}{V}{(E\ M)} \mid
  \refssl{r}{V}{(M\ E)}\mid{} $ \\
& & $\refssd{r}{V}{E}
  \mid \refssu{r}{V}{E}$ \\
$C$ & $::=$ & $[.] \mid (C \parallel M) \mid (M \parallel C)$ \\
$\mathbf{S}$ & $::=$ & $[.] \mid \mathbf{S} + \mathbf{M} \mid \mathbf{M} +
\mathbf{S}$ \mynl
\end{tabular}
\caption{Evaluation Contexts}\label{fig:lth-eval-context}
\end{minipage}
\end{table}

\begin{table}[!p]
\begin{center}
\begin{tabular}{lrcl}
\hline
\multicolumn{4}{|c|}{{\bf (a)}~$\beta$-reduction} \\
\hline
$(\beta_v)$ & $\refssl{r}{U}{((\lambda x.M)\,V)}$ & $\rightarrow$ &
$\refssd{r}{U}{(\varss{\{x\mapsto V\}}{M})}$ \\
\\
\multicolumn{4}{l}{If $M \rightarrow M'$ with rule $(\beta_v)$, then
  $\mathbf{S}[C[E[M]]] \rightarrow \mathbf{S}[C[E[M']]]$}\\
\\
\hline
\multicolumn{4}{|c|}{{\bf (b)}~Variable Substitutions} \\
\hline
\rrulet{subst}{var} & $\varss{\sigma}{x}$
& $\to$ & $\sigma(x)$ if defined, or $x$ otherwise\\
\rrulet{subst}{unit} & $\varss{\sigma}{\ast}$
& $\to$ & $\ast$ \\
\rrulet{subst}{app} & $\varss{\sigma}{\refssl{}{V}{(M\,N)}}$ & $\to$
& $\refrawsl{}{\varss{\sigma}{\multv{V}}}{((\varss{\sigma}{M})\,(\varss{\sigma}{N}))}$ \\
\rrule{subst}{\lambda} & $\varss{\sigma}{(\lambda y.M)}$ & $\to$
  & $\lambda y.(\varss{\sigma}{M})$ \\
\rrulet{subst}{get} & $\varss{\sigma}{\get{r}}$ & $\to$
 & $\get{r}$ \\
\rrule{subst}{\parallel} & $\varss{\sigma}{(M \parallel M')}$ & $\to$
& $(\varss{\sigma}{M}) \parallel (\varss{\sigma}{M'})$ \\
\rrulet{subst}{subst-r} &
$\varss{\sigma}{( \refssd{r}{V}{M} )}$ & $\to$
& $\refsrawsd{}{\metas{\sigma}{\multv{V}}}{\varss{\sigma}{M}}$ \\
\rrulet{subst}{subst-r'} &
$\varss{\sigma}{( \refssu{r}{V}{M} )}$ & $\to$
& $\refsrawsu{}{\metas{\sigma}{\multv{V}}}{\varss{\sigma}{M}}$ \\
\rrulet{subst}{merge} & $\varss{\tau}{\varss{\sigma}{M}}$ & $\to$ &
$\varss{\sigma,\tau}{M}$
\mynl
\multicolumn{4}{l}{Congruence case: }\\
\multicolumn{4}{l}{If $M \rightarrow M'$ by any of the previous rules, then
  $\mathbf{S}[C[E[M]] \rightarrow \mathbf{S}[C[E[M']]]$} \\
\\
\hline
  \multicolumn{4}{|c|}{{\bf (c)}~Downward Reference Substitutions} \\
\hline
\rrulet{subst-r}{val} & $\refssd{}{V}{V}$
& $\to$ & $V$ \\
\rrule{subst-r}{\parallel} & $\refssd{}{V}{(M
\parallel M')}$ & $\to$
& $(\refssd{}{V}{M}) \parallel (\refssd{}{V}{M'})$ \\
\rrulet{subst-r}{subst-r'} &
$\refssd{r}{V}{\refssu{s}{U}{M}}$ & $\to$
& $\refssu{s}{U}{\refssd{r}{V}{M}}$ \\
 & & & \\
\rrulet{subst-r}{merge} &
$\refssd{}{V}{\refssd{}{U}{M}}$ & $\to$
& $\refsrawsd{}{\multv{U},\multv{V}}{M}$ \\
\rrulet{subst-r}{app} &
$\refssd{}{V}{\refssl{}{U}{(M\,N)}}$ & $\to$
& $\refsrawsl{}{\multv{U},\multv{V}}{(\refssd{}{V}{M})\ (\refssd{}{V}{N})}$ \mynl
\multicolumn{4}{l}{Congruence cases:}\\
\multicolumn{4}{l}{If $M \rightarrow M'$ by any of the previous rules, then
  $\mathbf{S}[C[E[M]] \rightarrow \mathbf{S}[C[E[M']]]$} \mynl
\multicolumn{4}{l}{Finally:}\\
\rrulet{subst-r}{get} &
$\mathbf{S}[C[E[\refssd{}{V}{\get{r}}]]]$ & $\to$
& $\mathbf{S}[C[E[\get{r}]]] + \sum_{V \in \multv{V}(r)}
C[E[V]]$ \\
\\
\hline
\multicolumn{4}{|c|}{{\bf (d)}~Upward Reference Substitutions} \\
\hline
\rrule{subst-r'}{\parallel} & $(\refssu{r}{V}{M}) \parallel N$
& $\to$ & $\refssu{r}{V}{(M \parallel (\refssd{r}{V}{N}))}$ \\
 & & & \\
\rrulet{subst-r'}{lapp} &
$\refssl{}{U}{((\refssu{}{V}{M})\,N)}$ & $\to$ &
$\refssu{}{V}{\refsrawsl{}{\multv{U},\multv{V}}{(M\,(\refssd{}{V}{N}))}}$ \\
\rrulet{subst-r'}{rapp} &
$\refssl{}{U}{(M\,(\refssu{}{V}{N}))}$ & $\to$ &
$\refssu{}{V}{\refsrawsl{}{\multv{U},\multv{V}}{((\refssd{}{V}{M})\,N)}}$ \mynl
\multicolumn{4}{l}{Congruence case:}\\
\multicolumn{4}{l}{If $M \rightarrow M'$ by any of the previous rules, then
  $\mathbf{S}[C[E[M]] \rightarrow \mathbf{S}[C[E[M']]]$} \mynl
\multicolumn{4}{l}{Finally:}\\
\rrule{subst-r'}{\top} & $\mathbf{S}[\refssu{r}{V}{M}]$
& $\to$ & $\mathbf{S}[M]$
\end{tabular}
\end{center}
\caption{Reduction Rules. These are closed under the structural rules
  of Table~\ref{fig:lth-struct-rules}}\label{fig:lth-red-rules}
\end{table}

\begin{notation}\label{not:context}\rm
  the contexts $E$, $C$ and $\mathbf{S}$ are defined in
  Table~\ref{fig:lth-eval-context}. The context $E$ stands for a usual
  call-by-value applicative context, $C$ picks a thread, while
  $\mathbf{S}$ picks a term in a non-deterministic sum. Note how
  $E$ does not enforce any reduction order on an application.
\end{notation}
\begin{notation}\label{not:subst}\rm
  Given a variable substitution $\sigma$ and a value $V$, we define
  the value $\metas{\sigma}V$ as follows:
  $\metas{\sigma}{(\lambda x . M)} = \lambda x .
  (\varss{\sigma}{M})$,
  $\metas{\sigma}{(x)} = \sigma(x)$ if $\sigma$ is defined at $x$ or
  $\metas{\sigma}{(x)} = x$ if not, and
  $\metas{\sigma}{(\ast)} = \ast$.
\end{notation}
\begin{notation}\label{not:subst2}\rm
  We use the
   notation $\multv{X} = \multv{V},\multv{W}$ for the juxtaposition of
   references substitutions. It is defined by
   $\multv{X}(r) =
   \multv{V}(r) + \multv{W}(r)$ if both are defined and where $+$ is the union of multisets,
   $\multv{V}(r)$ if only $\multv{V}$ is defined, and
   $\multv{W}(r)$ if only $\multv{W}$ is defined.
   We use the same notation for the
   composition of variable substitutions, defined by
   $(\sigma,\tau)(x) =
     \metas{\tau}{\sigma(x)}$ if both are defined,
     $\sigma(x)$ if only $\sigma$ is defined
     $\tau(x)$  if only $\tau$ is defined.
   Finally, we define
   $\metas{\sigma}{\multv{V}} : r \mapsto [ \metas{\sigma}{V_i} \mid V_i \in
   \multv{V}(r) ]$.
\end{notation}

We now give some explanations on the rules of Table~\ref{fig:lth-red-rules}.

\paragraph*{(a) $\beta$-reduction}
If one forgets the $\lambda$-substitution explained below in
Subsection (d), this set of rules encapsulates the call-by-value
behavior of the language: only values can be substituted in the body
of abstractions, and this happens within a thread in a call-by-value
applicative context.

\paragraph*{(b) Variable Substitutions}

A variable
substitution can be seen as a message emitted by a $\beta$-redex and
dispatched through the term seen as a tree. The substitution flows
from the redex downward the term-tree until it reaches the occurrence
of a variable. The occurrence is then replaced, or not, depending on
the variable to be substituted. The rules in
Table~\ref{fig:lth-red-rules}(a) are an operational formalization of
this step-by-step procedure.
Consider for example the reduction of the term
$(\lambda x . x\ y\ z)\ (\lambda x . x)$. The redex triggers with
$(\beta_v)$ the substitution of all occurrences of $x$ in what was the
body of the lambda-abstraction. The substitution goes down the
corresponding sub-term and performs the substitution when it reaches an
occurrence of $x$.
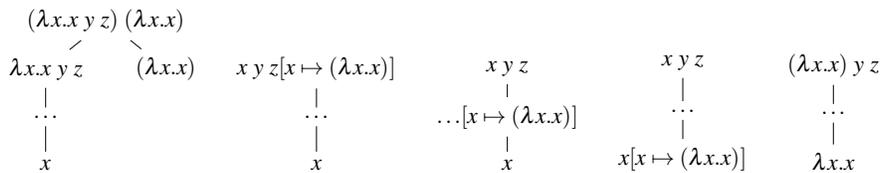
\begin{center}
\scalebox{.8}{%
\begin{tabular}{ccccc}
  \begin{tikzpicture}[yscale=\mytikshrink]
    \node[draw=none] (MN) at (0,0)
      {$(\lambda x . x\ y\ z)\ (\lambda x.x)$};
    \node[draw=none] (M) at (-1,-1)
      {$\lambda x. x\ y\ z$};
    \node[draw=none] (M2) at (-1,-2)
      {$\ldots$};
    \node[draw=none] (X) at (-1,-3)
      {$x$};
    \node[draw=none] (N) at (1,-1)
      {$(\lambda x. x)$};

    \draw[-,>=latex] (MN) -- (M);
    \draw[-,>=latex] (MN) -- (N);
    \draw[-,>=latex] (M) -- (M2);
    \draw[-,>=latex] (M2) -- (X);
  \end{tikzpicture}

  & \begin{tikzpicture}[yscale=\mytikshrink]
    \node[draw=none] (M) at (0,0)
    {$\varss{x \mapsto (\lambda x. x)}{x\ y\ z}$};
    \node[draw=none] (M2) at (0,-1)
      {$\ldots$};
    \node[draw=none] (X) at (0,-2)
      {$x$};

    \draw[-,>=latex] (M) -- (M2);
    \draw[-,>=latex] (M2) -- (X);
  \end{tikzpicture}

  & \begin{tikzpicture}[yscale=\mytikshrink]
    \node[draw=none] (M) at (0,0)
      {$x\ y\ z$};
    \node[draw=none] (M2) at (0,-1)
    {$\varss{x \mapsto (\lambda x. x)}{\ldots}$};
    \node[draw=none] (X) at (0,-2)
      {$x$};

    \draw[-,>=latex] (M) -- (M2);
    \draw[-,>=latex] (M2) -- (X);
  \end{tikzpicture}

  & \begin{tikzpicture}[yscale=\mytikshrink]
    \node[draw=none] (M) at (0,0)
      {$x\ y\ z$};
    \node[draw=none] (M2) at (0,-1)
      {$\ldots$};
    \node[draw=none] (X) at (0,-2)
      {$\varss{x \mapsto (\lambda x. x)}{x}$};
    \draw[-,>=latex] (M) -- (M2);
    \draw[-,>=latex] (M2) -- (X);
  \end{tikzpicture}

  & \begin{tikzpicture}[yscale=\mytikshrink]
    \node[draw=none] (M) at (0,0)
      {$(\lambda x. x)\ y\ z$};
    \node[draw=none] (M2) at (0,-1)
      {$\ldots$};
    \node[draw=none] (X) at (0,-2)
      {$\lambda x. x$};
    \draw[-,>=latex] (M) -- (M2);
    \draw[-,>=latex] (M2) -- (X);
  \end{tikzpicture}
\end{tabular}%
} 
\end{center}

\begin{remark}
On rules \rrulet{subst}{subst-r} and \rrulet{subst}{subst-r'}.
When composing or swapping substitutions, non-values may appear in
unfortunate places: take for example
$\varss{\sigma}{(\refssu{}{V}{\ast})}$, its reduction should be
$\refrawsu{}{\multv{V'}}{(\varss{\sigma}{\ast})}$ where
$\multv{V}'(r) = [ \varss{\sigma}{V} \mid V \in \multv{V}(r) ]$ when
$\multv{V}(r)$ is defined. But $\varss{\sigma}{V}$ are not necessarily
values and should then be able to be reduced inside
substitutions. However, note that that $\varss{\sigma}{V}$ always
reduces in one step to the value $\metas{\sigma}{V}$. To avoid
additional complexity, we perform this reduction at the same time,
whence the use of $\metas{\sigma}{V}$ instead of $\varss{\sigma}{V}$
in the actual rules.
\end{remark}

\paragraph*{(c) Downward Reference Substitutions}

An {\setw} can occur anywhere within a term and it must be able to reach a
{\getw} located in an arbitrary position. In a language such
as~\cite{Amadio2009}, as discussed in Section~\ref{sec:amadio}
the solution is to keep all assignments in a global
store. When a {\getw} gets evaluated, the value for the reference is
taken from the store.  This approach is very global in nature: the
store is ``visible'' by every subterm.

The language {{\lthis}} features a step-by-step decomposition of
reference assignments akin to term variable substitutions: an
assignment follows the branches of the term-tree, actively seeking a
{\getw}.
We therefore introduce two sorts of reference substitutions: one that
goes downward (indicated by $\downarrow$), similar to variable
substitutions, and one that goes upward (indicated by $\uparrow$). Starting from
an \setw, the latter climbs up the tree up to the root. The rules in
Table~\ref{fig:lth-red-rules}(c) describe the former while the rules in
Table~\ref{fig:lth-red-rules}(d) describe the latter.

\begin{remark}
  In Table~\ref{fig:lth-red-rules}, the rule \rrulet{subst-r}{get} is
  the central case of the reduction of reference substitutions. It
  says that whenever a downward substitution reaches a $\get{r}$, then
  it generates a non deterministic sum of all the available values for
  the reference $r$ (if $\multv{V}$ is undefined at $r$, then this sum
  is understood as a neutral element $\mathbf{0}$) plus a term where
  the substitution was discarded but the $\get{r}$ is left
  unreduced. To see why this ``remainder'' is necessary, consider the
  term $\refsrawsd{}{\store r V_2}{\refsrawsd{}{\store r V_1}{\get{r}}}$.  If we
  omit the remainder, the term could reduce to
  $\refrawsd{}{\store r V_2}{V_1}$ and finally to $V_1$. But another
  reduction is possible: one can first reduce the term to
  $\refsrawsd{}{\store{r}{V_1},\store{r}{V_2}}{\get{r}}$ and then to $V_1 + V_2$.  The
  $\get{r}$ must not be greedy: when it meets a substitution, it has
  to consider the possibility that other substitutions will be
  available later. This aspect will be crucial when considering the
  proof of strong normalization of the language in
  Section~\ref{sec:term}.
\end{remark}

\paragraph*{(d) Upward Reference Substitutions}

Each time an upward reference substitution goes through a multi-ary constructor --
as an application or a parallel composition -- it propagates downward
substitutions in all the children of the constructor except the one
it comes from, while continuing its ascension. Eventually, all the
leafs are reached by a corresponding downward substitution. To
illustrate the idea, consider a term $M\ N$ where $M$ contains a
$\get{r}$ somewhere and $N$ an {\setw} $\refrawsu{}{\store r V}{\ast}$. The reduction of explicit
substitutions would go as follows.

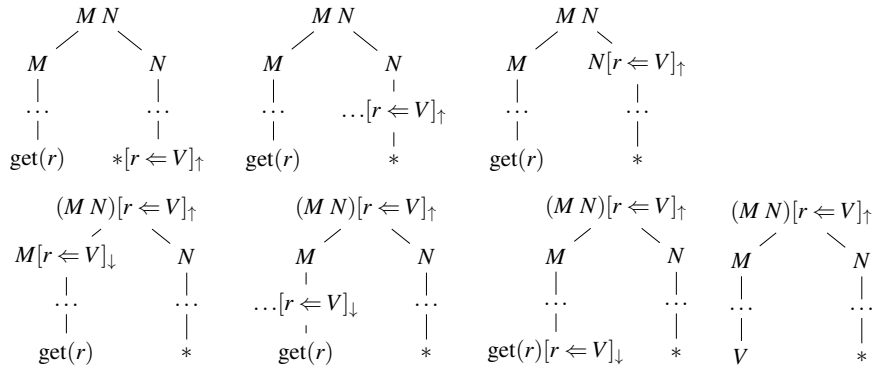
\begin{center}
\scalebox{.8}{%
\begin{tabular}{cccc}
  %
  %
  \begin{tikzpicture}[yscale=\mytikshrink]
    \node[draw=none] (MN) at (0,0)
      {$M\ N$};
    \node[draw=none] (M) at (-1,-1)
      {$M$};
    \node[draw=none] (M2) at (-1,-2)
      {$\ldots$};
    \node[draw=none] (G) at (-1,-3)
      {$\text{get}(r)$};
    \node[draw=none] (N) at (1,-1)
      {$N$};
    \node[draw=none] (N2) at (1,-2)
      {$\ldots$};
    \node[draw=none] (S) at (1,-3)
      {$\refrawsu{}{\store r V}{\ast}$};

    \draw[-,>=latex] (MN) -- (M);
    \draw[-,>=latex] (MN) -- (N);
    \draw[-,>=latex] (M) -- (M2);
    \draw[-,>=latex] (M2) -- (G);
    \draw[-,>=latex] (N) -- (N2);
    \draw[-,>=latex] (N2) -- (S);
  \end{tikzpicture}

  & \begin{tikzpicture}[yscale=\mytikshrink]
    \node[draw=none] (MN) at (0,0)
      {$M\ N$};
    \node[draw=none] (M) at (-1,-1)
      {$M$};
    \node[draw=none] (M2) at (-1,-2)
      {$\ldots$};
    \node[draw=none] (G) at (-1,-3)
      {$\text{get}(r)$};
    \node[draw=none] (N) at (1,-1)
      {$N$};
    \node[draw=none] (N2) at (1,-2)
      {$\refrawsu{}{\store r V}{\ldots}$};
    \node[draw=none] (S) at (1,-3)
      {$\ast$};

    \draw[-,>=latex] (MN) -- (M);
    \draw[-,>=latex] (MN) -- (N);
    \draw[-,>=latex] (M) -- (M2);
    \draw[-,>=latex] (M2) -- (G);
    \draw[-,>=latex] (N) -- (N2);
    \draw[-,>=latex] (N2) -- (S);
  \end{tikzpicture}

  & \begin{tikzpicture}[yscale=\mytikshrink]
    \node[draw=none] (MN) at (0,0)
      {$M\ N$};
    \node[draw=none] (M) at (-1,-1)
      {$M$};
    \node[draw=none] (M2) at (-1,-2)
      {$\ldots$};
    \node[draw=none] (G) at (-1,-3)
      {$\text{get}(r)$};
    \node[draw=none] (N) at (1,-1)
      {$\refrawsu{}{\store r V}{N}$};
    \node[draw=none] (N2) at (1,-2)
      {$\ldots$};
    \node[draw=none] (S) at (1,-3)
      {$\ast$};

    \draw[-,>=latex] (MN) -- (M);
    \draw[-,>=latex] (MN) -- (N);
    \draw[-,>=latex] (M) -- (M2);
    \draw[-,>=latex] (M2) -- (G);
    \draw[-,>=latex] (N) -- (N2);
    \draw[-,>=latex] (N2) -- (S);
  \end{tikzpicture} \\

  \begin{tikzpicture}[yscale=\mytikshrink]
    \node[draw=none] (MN) at (0,0)
      {$\refrawsu{}{\store r V}{(M\ N)}$};
    \node[draw=none] (M) at (-1,-1)
      {$\refrawsd{}{\store r V}{M}$};
    \node[draw=none] (M2) at (-1,-2)
      {$\ldots$};
    \node[draw=none] (G) at (-1,-3)
      {$\text{get}(r)$};
    \node[draw=none] (N) at (1,-1)
      {$N$};
    \node[draw=none] (N2) at (1,-2)
      {$\ldots$};
    \node[draw=none] (S) at (1,-3)
      {$\ast$};

    \draw[-,>=latex] (MN) -- (M);
    \draw[-,>=latex] (MN) -- (N);
    \draw[-,>=latex] (M) -- (M2);
    \draw[-,>=latex] (M2) -- (G);
    \draw[-,>=latex] (N) -- (N2);
    \draw[-,>=latex] (N2) -- (S);
  \end{tikzpicture}

  & \begin{tikzpicture}[yscale=\mytikshrink]
    \node[draw=none] (MN) at (0,0)
      {$\refrawsu{}{\store r V}{(M\ N)}$};
    \node[draw=none] (M) at (-1,-1)
      {$M$};
    \node[draw=none] (M2) at (-1,-2)
      {$\refrawsd{}{\store r V}{\ldots}$};
    \node[draw=none] (G) at (-1,-3)
      {$\text{get}(r)$};
    \node[draw=none] (N) at (1,-1)
      {$N$};
    \node[draw=none] (N2) at (1,-2)
      {$\ldots$};
    \node[draw=none] (S) at (1,-3)
      {$\ast$};

    \draw[-,>=latex] (MN) -- (M);
    \draw[-,>=latex] (MN) -- (N);
    \draw[-,>=latex] (M) -- (M2);
    \draw[-,>=latex] (M2) -- (G);
    \draw[-,>=latex] (N) -- (N2);
    \draw[-,>=latex] (N2) -- (S);
  \end{tikzpicture}

  & \begin{tikzpicture}[yscale=\mytikshrink]
    \node[draw=none] (MN) at (0,0)
      {$\refrawsu{}{\store r V}{(M\ N)}$};
    \node[draw=none] (M) at (-1,-1)
      {$M$};
    \node[draw=none] (M2) at (-1,-2)
      {$\ldots$};
    \node[draw=none] (G) at (-1,-3)
      {$\refrawsd{}{\store r V}{\get r}$};
    \node[draw=none] (N) at (1,-1)
      {$N$};
    \node[draw=none] (N2) at (1,-2)
      {$\ldots$};
    \node[draw=none] (S) at (1,-3)
      {$\ast$};

    \draw[-,>=latex] (MN) -- (M);
    \draw[-,>=latex] (MN) -- (N);
    \draw[-,>=latex] (M) -- (M2);
    \draw[-,>=latex] (M2) -- (G);
    \draw[-,>=latex] (N) -- (N2);
    \draw[-,>=latex] (N2) -- (S);
  \end{tikzpicture}

  & \begin{tikzpicture}[yscale=\mytikshrink]
    \node[draw=none] (MN) at (0,0)
      {$\refrawsu{}{\store r V}{(M\ N)}$};
    \node[draw=none] (M) at (-1,-1)
      {$M$};
    \node[draw=none] (M2) at (-1,-2)
      {$\ldots$};
    \node[draw=none] (G) at (-1,-3)
      {$V$};
    \node[draw=none] (N) at (1,-1)
      {$N$};
    \node[draw=none] (N2) at (1,-2)
      {$\ldots$};
    \node[draw=none] (S) at (1,-3)
      {$\ast$};

    \draw[-,>=latex] (MN) -- (M);
    \draw[-,>=latex] (MN) -- (N);
    \draw[-,>=latex] (M) -- (M2);
    \draw[-,>=latex] (M2) -- (G);
    \draw[-,>=latex] (N) -- (N2);
    \draw[-,>=latex] (N2) -- (S);
  \end{tikzpicture}
\end{tabular}%
} 
\end{center}

One last subtlety in the movement of reference substitutions concerns
$\lambda$-abstractions. As made explicit in
Table~\ref{fig:lth-red-rules}(a), the language is call-by-value:
reduction does not happen under $\lambda$-abstractions. In particular, a
read within the body of a $\lambda$-abstraction should only be accessible
by an assignment when the $\lambda$-abstraction is opened: we have a
natural notion of pure and impure terms. Pure terms are terms that
will not produce any effect when reduced, and in particular, all values
are expected to be pure terms since they cannot reduce further. This is
highlighted by rule \rrulet{subst-r}{val}: when encountering a pure
term, a reference substitution vanishes. But the case of abstraction
is more subtle: computational effects frozen in its body are freed
when the abstraction is applied. If one implements naively the
reduction rules of reference substitutions, then the following example
does not behave as expected:
$\refssd{}{V}{((\lambda x. \get{r})\ \ast)} \to
(\refssd{}{V}{(\lambda x .  \get{r})})\ (\refssd{}{V}{\ast}) \to
(\lambda x .  \get{r})\ \ast \to \get{r}$.
We end up with an orphan $\get{r}$ despite the fact that a substitution was
available at the beginning.
The problem is that the substitution diffuses through the application,
then encounters two pure terms and vanishes.

In an application, the left term eventually exposes the body of an
abstraction, and this body should be able to use any substitution that
was in its scope. The $\lambda$-substitution $[-]_\lambda$ is a
special stationary reference substitution attached to an
application. Its goal is precisely to record all the substitutions
that went down through it with Rules \rrulet{subst-r'}{lapp} and
\rrulet{subst-r'}{rapp}. When the application is finally reduced with
a $\beta_v$-rule, this substitution will turn to a downward one and
feed the $\get{r}$'s that were hidden in the abstraction's
body.

\begin{remark}
  An alternative approach to $\lambda$-substitution would be to make
  downward substitutions not vanish (i.e. getting rid of Rule
  \rrulet{subst-r}{val}). In this situation, values would be handled
  with their whole context of references assignment. Apart from the
  heavy syntactical cost of carrying around a lot of similar and
  possibly useless substitutions, the idea that hidden effects are
  released at application appears more natural regarding type and
  effect systems, as the one we introduce in Section~\ref{sec:type}.
\end{remark}

\begin{remark}
  Rule \rrule{subst-r'}{\top} acts as a garbage
  collection to eliminate top-level upward substitutions. While not
  necessary, this will greatly ease the statement and proof of lemmas
  and theorems (such as Lemma~\ref{lemma-progress}).
\end{remark}




\section{Stratification and Type System}
\label{sec:type}

We present in this section a stratified type and effect system for {\lthis}
inspired from~\cite{Amadio2009,madet:tel-00794977}.
A type and effect system aims at statically track the potential
effects that a term can produce when reduced. Here, the considered effects are
read from or write to references.

\subsection{The Type System of {{\lthis}}}

\newcommand{\Unit}{\ensuremath{\texttt{Unit}}}

Formally, the type and effect system is defined as follows.
\begin{center}
\begin{tabular}{lrcl}
  -effects & $e,e'$ & $\subset$ & $\{r_1,r_2,\ldots\}$ \\
  -types & $\alpha$ & $::=$ & $\mathbf{B} \mid A$ \\
  -value types & $A$ & $::=$ & $\Unit \mid A \sur{\to}{e} \alpha
  \mid \text{Ref}_r A$
\end{tabular}
\end{center}
The type $\Unit$ is the type of $\ast$. The function type
$A \sur{\to}{e} \alpha$ is annotated with an effect $e$: the set of
references the function is allowed to use. Finally, the type
$\text{Ref}_r A$ states that the reference $r$ can only be substituted
with values of type $A$.
Since thread cannot be fed as an argument to a function, the type of the
parallel components of a program is irrelevant. They are given the
opaque behavior type $\mathbf{B}$.
We separate $\alpha$-types and $A$-types to ensure
that $\mathbf{B}$ cannot be in the domain of a function.

In the typing rules we use two distinct contexts: variable contexts
$\Gamma$ of the form $x_1 : A_1,\ldots,x_n : A_n$ and reference
contexts $R$ of the form $r_1 : A_1,\ldots,r_n : A_n$. The latter
indicates the type of the values that a reference $r$ appearing in $M$
can hold. If the order of variables in $\Gamma$ is irrelevant, the
order of references in $R$ is important.

In order to ensure termination, the type and effect system is
stratified: this stratification induces an order forbidding circularity
in reference assignments. It is presented as a set of rules to build
the reference context and can be found in Figure~\ref{fig:strat}. It
states that when a new reference is added to the context,
all references appearing in its type must already be in $R$.
In Figure~\ref{fig:strat} the entailment symbol $(\vdash)$
is overloaded with several meanings:
\begin{itemize}
  \item $R$ is well formed, written {$R \vdash$}, means that the references appearing in
    $R$ are stratified.
  \item A type $\alpha$ is well formed under $R$, written {$R \vdash (\alpha,e)$}, means
    that all references appearing in $e$ and $\alpha$ are in $R$.
  \item A variable context $\Gamma$ is well formed under $R$, written $R \vdash \Gamma$,
    means that all the types appearing in $\Gamma$ are well formed under $R$.
\end{itemize}

The type and effect system features a subtyping relation whose
definition rules are presented in Figure \ref{fig:subtyp}. It
formalizes the idea that a function of type $A \sur{\to}{\{r\}} \alpha$ is
not obliged to use the reference $r$.

Typing judgments overload once more the symbol $(\vdash)$ and take
the form $R; \Gamma \vdash M : (\alpha,e)$ where $R$ is the reference
context, $\Gamma$ the variable context, $\alpha$ the type of $M$ and
$e$ the references that $M$ may affect.  Using the stratification and
the subtyping relation, the typing rules for the language {\lthis} are
presented in Figure~\ref{lces-typ-rules}. For succinctness, the application rule
has been factorized into two rules, (APP) and 
(SUBST) for $\xi = \lambda$. Thus (APP) is not a legitimate rule but an abuse of notation, and
must be followed by an appropriate (SUBST) in any type derivation.

\begin{remark}
In Rule \textsc{(lam)}, when abstracting over a variable in a term
$R;\Gamma,x:A \vdash M : (\alpha,e)$, the resulting value
$\lambda x. M$ is pure and hence its effects should be the empty
set. However one must remember that the body of this abstraction is
potentially effectful: this is denoted by annotating the functional
arrow ``$\to$'' with a superscript indicating these effects.
Also note that in general, the
order of references in $R$ is capital: it is the order
induced by stratification.
\end{remark}

\begin{figure}[tb]
\begin{gather*}
\infer[]
{}{\emptyset \vdash}
\qquad
\infer[]
{ R \vdash A \qquad r \notin \text{dom}(R) }
{ R,r:A \vdash }
\qquad
\infer[]
{ R \vdash }
{ R \vdash \Unit }
\qquad
\infer[]
{ R \vdash }
{ R \vdash \mathbf{B}}
\mynl
\infer[]
{ R \vdash A \qquad R \vdash \alpha \qquad e \subseteq \text{dom}(R)}
{ R \vdash A \sur{\rightarrow}{e} \alpha }
\qquad
\infer[]
{ R \vdash \qquad r : A \in R }
{ R \vdash \text{Ref}_r A }
\end{gather*}
\caption{Stratification of the type system}\label{fig:strat}
\end{figure}

\begin{figure}[tb]
\begin{gather*}
\infer[(ref)]
{}{R \vdash \alpha \leq \alpha}
\qquad
\infer[(arrow)]
{ R \vdash A' \leq A \qquad R \vdash (\alpha,e) \leq (\alpha',e') }
{ R \vdash A \sur{\to}{e} \alpha \leq A' \sur{\to}{e'}
\alpha' }
\qquad
\infer[(cont)]
{ e \subset e' \subset \text{dom}(R) \qquad R \vdash \alpha \leq \alpha' }
{ R \vdash (\alpha,e) \leq (\alpha',e') }
\end{gather*}
\caption{Subtyping relation}\label{fig:subtyp}
\end{figure}

\begin{figure}[tb]
\begin{gather*}
\infer[(var)]
{R \vdash \Gamma,x : A}
{R;\Gamma,x : A \vdash x : (A,\emptyset)}
\qquad
\infer[(unit)]
{R \vdash \Gamma} {R;\Gamma \vdash \ast : (\Unit,\emptyset)}
\qquad
\infer[(reg)]
{R \vdash \Gamma \qquad r : A \in R}
{R;\Gamma \vdash r : Ref_r{A}}
\mynl
\infer[(lam)]
{ R;\Gamma,x : A \vdash M : (\alpha,e) }
{ R;\Gamma \vdash \lambda{x}.M : (A \sur{\to}{e} \alpha,\emptyset) }
\qquad
\infer[(app)]
{ R;\Gamma \vdash M : (A \sur{\to}{e_1} \alpha, e_2) \qquad R;\Gamma \vdash N : (A,e_3) }
{ R; \Gamma\vdash M\ N : (\alpha, e_1 \cup e_2 \cup e_3) }
\mynl
\infer[(get)]
{ R;\Gamma \vdash \text{Ref}_r{A} }
{ R;\Gamma \vdash \get{r} : (A,\{ r \}) }
\qquad
\infer[(sub)]
{ R;\Gamma \vdash M : (\alpha,e) \qquad R \vdash (\alpha,e) \leq (\alpha',e') }
{ R;\Gamma \vdash M : (\alpha',e') }
\mynl
\infer[(subst)]
{ R; \Gamma, x_1 : A_1, \ldots, x_n : A_n \vdash M : (\alpha,e) \qquad
  \forall i:\ \ R; \Gamma \vdash V_i : (A_i,\emptyset) }
  { R; \Gamma \vdash \varss{\forall i:\, x_i \mapsto V_i}{M} : (\alpha, e) }
\mynl
\infer[(subst\text{-}r)]
{ \forall i:\ R; \Gamma \vdash r_i : \text{Ref}_{r_i} A_i \quad
  R; \Gamma \vdash M : (\alpha,e) \quad
  \forall i:\ r_i \in e \quad
  \forall i:\ V \in \mathscr{E}_i \implies R; \Gamma \vdash V : (A_i,\emptyset) }
{ R; \Gamma \vdash M[\forall i:\, r_i \mapsto \mathscr{E}_i]_{_\xi} :
  (\alpha, e) \qquad  \text{for } \xi \in \{\uparrow,\downarrow,\lambda\} }
\mynl
\infer[(par)]
{i=1,2 \qquad R;\Gamma \vdash M_i : (\alpha_i,e_i) }
{ R;\Gamma \vdash M_1 \parallel M_2 : (\mathbf{B},e_1 \cup e_2) }
\qquad
\infer[(sum)]
{i=1,2 \qquad R;\Gamma \vdash \mathbf{M_i} : (\alpha,e) }
{ R;\Gamma \vdash \mathbf{M}_1 + \mathbf{M}_2 : (\alpha,e) }
\end{gather*}
\caption{Typing rules for $\lambda_\text{cES}$}\label{lces-typ-rules}
\end{figure}


\subsection{Basic Properties of $\lambda_\text{cES}$ }

The language $\lambda_\text{cES}$ satisfies the usual safety
properties of a typed calculus.
First, {\lthis} enjoys subject reduction.

\begin{lemma}[Subject reduction]\label{lemma-subject-reduction}
  Let $R;\Gamma \vdash M : (\alpha,e)$ be a typing judgment, and assume that $M
  \rightarrow M'$. Then $R;\Gamma \vdash M' : (\alpha,e)$.\qed
\end{lemma}

\begin{remark}
  The fact that an effectful term
may become pure after reduction is reflected by the subtyping
relation.  For example, consider $P = \refssd{}{V}{\get{r}}$ where
$R \vdash P : (A,\{r\})$ and $P \to (V + \ldots)$. Since $V$ is a
value it can only be given the type $R \vdash V : (A,\emptyset)$.
Subject reduction would however require that $V$ has the same type
$(A,\{r\})$ as $P$. The subtyping relation corresponds to effect
containment, meaning that the effects appearing in types are an upper
bound of the actual effects produced by a term, so that
$(A,\emptyset)$ is a subtype of $(A,\{r\})$.
\end{remark}

Well-typed normal forms of {\lthis} may not be values. For example,
the term $\get{r}$ is not a value.  The progress theorem states that the
only reason for which a term may get stuck is the presence of an orphan {\getw}
with no corresponding {\setw}. Normal forms are thus either values, or
some application of values together with at least one such stuck
{\getw}.
\begin{lemma}[Progress]\label{lemma-progress}
  Let $R \vdash \mathbf{M} : (A,e)$ be a typable program that does not reduce
  further. Then $\mathbf{M}$ is of the form
 $\sum_{i=1}^n (M^i_1 \parallel \ldots \parallel M^i_{l_i})$ where the
 $M^i_j$ are
  either values or terms of the grammar
  $ M_{\text{norm}} ::= \get{r} \mid
  \refssl{r}{V}{(M_{\text{norm}}\ V)} \mid \refssl{r}{V}{(V\
    M_{\text{norm}})} \mid \refssl{r}{V}{(M_{\text{norm}}\
    M_{\text{norm}})}$.
  \qed
\end{lemma}

\newcommand{\rednd}{\to_\texttt{nd}}
\section{Termination}
\label{sec:term}

Our main result is a finitary, interactive proof of strong
normalization for {\lthis}. This section is devoted to the presentation
of the problem in the context of references, the explanation of
why the existing solutions do not apply to our setting and what we
propose instead.

\subsection{Shortcoming of Existing Solutions}

Introduced by Tait in 1967~\cite{tait67symboliclogjkic},
reducibility is a widely used,
versatile technique for proving strong normalization of
lambda-calculi. The core of this technique is to define inductively on
types $\tau$ a set $\mathbf{SC}(\tau)$ of well typed terms, called
{\em strongly computable terms}, satisfying a series of
properties. One proves that terms in $\mathbf{SC}(\tau)$ are strongly
normalizing (Adequacy) and (the most difficult part) that all well
typed terms of a type $\tau$ are actually in $\mathbf{SC}(\tau)$.

When adapting this technique to a type and effect system, the main
difficulty is that the definition is not obviously inductive anymore.
To define $\mathbf{SC}{(A \sur{\to}{e} \alpha)}$, we need to have
defined the types of references appearing in $e$. But $e$ can itself contain
a reference of type $A \sur{\to}{e} \alpha$: in the Landin's
fixpoint example shown in Section~\ref{sec:intro}, the looping term
has the type $\mbox{\Unit} \sur{\to}{\{r\}} \mbox{\Unit}$ while $r$ has
the same type. The role of stratification is to induce a
well-founded ordering on types so that the definition becomes
consistent.

The solution offered by stratification of the type system is however
not enough for {\lthis}. In
Boudol~\cite{BOUDOL2010716} where the technique is introduced,
concurrency is explicitly
controlled by threads themselves that are guaranteed to be the only process in
execution during each slice of execution. In {\lthis}, reduction
steps are performed in arbitrary threads such that stores may be affected by
others between two atomic steps in a particular thread.

To overcome this issue, for the language presented in Section~\ref{sec:amadio},
Amadio~\cite{Amadio2009} strengthens the condition defining $\mathbf{SC}$ sets
by asking that they also terminate under infinite stores of the form $(\store{r}{V_1}
\parallel \ldots \parallel \store{r}{V_n} \parallel \ldots)$ with $\seq{V}$
an enumeration of all the elements of
$\mathbf{SC}(\alpha)$. In this setting,
infinite stores are static top-level constructions: once saturated,
they are invariant by any new {\setw}. In a term $(M_1 \parallel M_2 \parallel S)$
with $S$ being such a store, any memory operation of $M_2$ is completely
invisible to $M_1$ and one can prove separately the termination of each thread.

However, this solution is not easily transposable to {\lthis}.
First of all, the rule \rrulet{subst-r}{get} produces all the possible
values associated to a store. The corresponding $\refssd{}{V}{\get{r}}$ would
reduce to an infinite sum $\get{r} + \sum_i V_i$ where, even if each summand
terminates, there is also for any positive integer $n$ a summand that takes at
least $n$ steps to reach normal form. The total sum is not terminating anymore.
Secondly, unlike static top-level stores, reference substitutions are
duplicated, erased and exchanged in an interactive way between
threads.

\subsection{Our Solution}
\label{sec:our}

To prove strong normalization of {\lthis}, we change gears.
With explicit substitutions, \setw s\
and \getw s\ are a way of exchanging messages between threads or subterms. Apart
from the termination of each term in isolation, the key property we need
is that threads cannot exchange an infinite amount of messages.

We formalize this condition by strengthening the definition of
strongly computable terms. We force them to also be {\em
  well-behaved}. A well-behaved term must only emit a finite number of
upward substitutions containing strongly computable terms when placed
in a ``fair'' context. A fair context is a context that would only
send strongly computable reference substitutions (albeit potentially
infinitely many).

These notions are defined in Section~\ref{sec:tech-defs}, while the
strong normalization result is spelled out in
Section~\ref{sec:sn}.

\subsection{Technical Definitions}
\label{sec:tech-defs}

\begin{remark}
  In the following, we do not want to deal with the clumsiness of handling sums
  of terms everywhere. If a reduction sequence is seen as a tree, where branching points
  correspond to \rrulet{subst-r}{get} and the children to all the summands
  produced by this rule, then by K\"{o}nig's lemma it is finite if and only if all its
  branches are finite. We will thus use an alternative
  non-deterministic reduction, denoted by $\rednd$, such that a sequence of
  reductions~$\rednd$ corresponds to a branch in the original
  reduction system. The
  termination of $\rednd$ is sufficient, thanks to the above remark.
  We define $\rednd$ by replacing the \rrulet{subst-r}{get} reduction by the
  following two rules:
  \begin{tabred}
    \refssd{}{V}{\get{r}} & \rednd & V \text{ if } V \in \multv{V}(r) \\
    \refssd{}{V}{\get{r}} & \rednd & \get{r}
  \end{tabred}
  In the rest of the paper, we only consider simple terms (non-sums) and the $\rednd$
  reduction.
\end{remark}

The purpose of the following Definition~\ref{definition-env-reduction}
is to formalize the interaction of a subterm with its context as a play against an opponent that
can non-deterministically drop downward substitutions at the top level or absorb upcoming
substitutions. This is summarized in the condition {\condWB} of
Definition~\ref{definition-strongly-computable}.

\begin{definition}[Environment Reduction]%
  \label{definition-env-reduction}%
  \rm
  Let $\vdash M : (\alpha,e)$ be a well typed term. Let $\seq{\multv{V}}$
  be a sequence of reference substitutions such that
  $\refrawsd{}{\multv{V}_i}{M}$ is well typed: we denote it with
  $\vdash (M,\seq{\multv{V}})$.  We call a $(M,\seq{\multv{V}})$-reduction a finite
  or infinite reduction sequence starting from $M$ where each step is either a
  $\rednd$, or an interaction with the environment defined by the additional
  rules $\refssu{}{V}{M} \to_\uparrow M$ and $M \to_\downarrow \refrawsd{}{\multv{V}_i}{M}$.
\end{definition}

We define the notion of strongly computable terms discussed in
Section~\ref{sec:our} as follows.

\begin{notation}
  By abuse of notation, in
  Definition~\ref{definition-strongly-computable}
  we shall omit the $R$ or $(\alpha,e)$ when it is
  obvious from the context and just write $M \in
  \mathbf{SC}$.
  Moreover, we abusively write $\multv{V} \subseteq \mathbf{SC}$ to
  mean that for all $r$ where $\multv{V}$ is defined we have
  $\multv{V}(r) \subseteq \mathbf{SC}_R(R(r),\emptyset)$.
\end{notation}

\begin{definition}[Strongly Computable Terms]%
  \label{definition-strongly-computable}%
  \rm The set $\mathbf{SC}_R(\alpha,e)$ of strongly computable terms
  of type $(\alpha,e)$ is defined by induction on the type $\alpha$.

\smallskip
\noindent
{\em Base case.}~~
Assume that $\alpha = \Unit \mid \mathbf{B}$ and that $R \vdash M :
    (\alpha,e)$. Then  $M \in \mathbf{SC}_R(\alpha,e)$ if it is
    \begin{description}
      \item[\condSN]Strongly normalizing under reference
        substitutions: For all $\multv{V} \subseteq \mathbf{SC}_R$,
         the term $\refssd{r}{V}{M}$ is strongly normalizing.
      \item[\condWB]Well Behaved:
        For any $\seq{\multv{V}}$ with $\forall i: \multv{V}_i
        \subseteq \mathbf{SC}_R$ and $\vdash (M,\seq{\multv{V}})$,
        for any $(M,\seq{\multv{V}})$-reduction $M = M_0 \to\!\ldots\!\to M_n \to
        \ldots$, there exists $n_0 \geq 1$ such that for all $k \geq 1$:
        \begin{enumerate}
          \item If $M_{k-1}$ is of the form $\refssu{}{U}{N}$ with $M_{k-1} \to_\uparrow M_k$
            then $\multv{U} \subseteq \mathbf{SC}$,
          \item If $k \geq n_0$ then $M_{k-1} \to M_k$ is not a $(\to_\uparrow)$
            step.
        \end{enumerate}
    \end{description}
\smallskip
\noindent
{\em Inductive case.}~~ Assume that $R \vdash M :
    (A \sur{\to}{e_1} \alpha,e)$ with $e_1 \subseteq e$. Then $M$ belongs to $\mathbf{SC}_R(A \sur{\to}{e_1} \alpha,e)$
    provided that
    for all $N \in \mathbf{SC}_R(A,e)$, we have $M\,N \in \mathbf{SC}_R(\alpha,e)$.
    The requirement that $e_1 \subseteq e$ can always be assumed
    without loss of generality thanks to subtyping.
\end{definition}

\begin{remark}
  The condition {\condSN} requires terms to be strongly normalizing when
  put under any finite reference substitution of strongly computable
  terms. The finiteness is sufficient, thanks to the presence of
  condition {\condWB}. This rather technical condition is the
  well-behaved requirement developed in Section~\ref{sec:our}: it says
  that there are at most $n_0$ ($\to_\uparrow$) steps.
\end{remark}

In the proof of termination we make use of a preorder $\rincl$. The proposition $M
\rincl N$ means that the two terms are essentially the same, but that $N$ may
have more available {\setw}s, and possibly in different positions.  This is
typically the case if $N$ is a reduct of $\refssd{}{V}{M}$ using only downward
structural rules. We give the full definition and important properties in Appendix~\ref{ap:termination}.
For the purpose of the proof of termination, its interesting property is the
following one:

\begin{lemma}[Simulation Preorder]\label{lemma-ord-simulation}
  Assume that $M \rincl N$. If $N$ is strongly normalizing then $M$
  is strongly normalizing.\qed
\end{lemma}

\subsection{Strong Normalization for {\lthis}}
\label{sec:sn}

We are now ready to state and sketch the proof of strong-normalization for {\lthis}.
The easy part is the adequacy result, stated as follows.

\begin{lemma}[Adequacy]\label{lemma-adequacy}
  If $M \in \mathbf{SC}(\alpha,e)$ then $M$ is strongly normalizable.\qed
\end{lemma}

The heart of our result is the opposite result, the soundness:

\begin{lemma}[Soundness]\label{lemma-soundness}
  Suppose that  $R; x_1:A_1,\ldots,x_n:A_n \vdash P : (\alpha,e)$,
  and that
  $\sigma$ maps each $x_i$ to some $V_i\in\mathbf{SC}(A_i)$. Then
  $\varss{\sigma}{P} \in \mathbf{SC}(\alpha,e).$
\end{lemma}

%
\noindent
{\em Sketch of the proof of Lemma \ref{lemma-soundness}.}
The proof is performed by induction on the structure of the term $P$. To show
how the proof works, we focus on a representative case.

Let us treat the case $P = \varss{\sigma}{(\lambda x. M)}$. We assume that $M
\in \mathbf{SC}$, and we want to show that $P \in \mathbf{SC}$. Let $\alpha =
A_1 \sur{\to}{e_1} \ldots \sur{\to}{e_{n-1}} A_n \sur{\to}{e_n} \beta$ be the
expansion of the type of $P$, where $\beta$ is either $\mbox{\Unit}$ or $\mathbf{B}$.
If we unfold the recursive definition of $\mathbf{SC}$ sets, proving that $P \in
\mathbf{SC}(\alpha)$ amounts to check that
$\Lambda(P,\multv{U},\seq{N}) := \refssd{}{U}{P\ N_1\ \ldots\ N_n}$
satisfies {\condSN} and {\condWB} for all strongly computable
$N_1,\ldots,N_n$ and $\multv{U}$ with suitable types.
By abuse of notation we omit some
parameters of $\Lambda(P,U,(N_i))$ and write $\Lambda(P)$ when clear. If $\Lambda(P) \to^\ast P'\ N_1'\ \ldots\
N_n'$ (omitting some reference substitutions) where $\beta_V$ reductions occur
only inside subterms $P,N_1,\ldots,N_n$, we
will liberally call $N_i'$ a reduct of $N_i$. We focus on the strong
normalization of $\Lambda(P)$, the well-behaved condition being proved
in a similar manner.
Consider a sequence $S$ of reductions of $\Lambda(P)$. The head term $P$ has
only one possible reduction, namely  $P' = \lambda x . (\varss{\sigma}{M})$, and must
then take part in a $(\beta_v)$ reduction with a reduct of $N_1$ to reduce
further. We consider the two cases:
\begin{description}
  \item[P is passive] If such a $(\beta_v)$ does not occur in the
    sequence $S$, all the
    reducts of $\Lambda(P)$ have the form $\Lambda' = Q\ N_1'\ \ldots\ N_n'$ where
    $Q$ is either $P$ or $P'$ and where for all $i$,
    $N_i'$ is a reduct of $N_i$. We omitted a bunch of floating reference
    substitutions for the sake of readability.  $Q$ is inert and does not play
    any role in the termination: we can focus on showing that all the reducts 
    of $N_i$s cannot diverge. While they do terminate in isolation as strongly
    computable terms, the possibility of an infinite exchange of substitutions
    prevent us from using {\condSN} directly. This is the precise role of {\condWB}:
    the reduction of each $N_i$ can be mapped to an environment reduction. We
    adopt the following strategy :
    \begin{enumerate}
      \item Use {\condWB} to show that the exchange of substitutions must come to
        an end
      \item For each $N_i'$, gather all the substitutions (a finite number
        according the previous step) it receives during the reduction of
        $\Lambda(P)$ and merge them into one $\multv{X}_i$
      \item Show that we can bound each reduct $N_i'$ : $N_i' \rincl
        \refsrawsd{}{\multv{X}_i}{N_i}$
    \end{enumerate}
    Since the bounding terms are strongly normalizing by {\condSN}, so are the
    $N_i'$s by Lemma~\ref{lemma-ord-simulation}, and the considered reduction is finite.
  \item[P is active] Now, assume that at some point the leftmost application
    $\refssl{}{W}{Q\ N_1'}$ is reduced to $Q' =
    \refssd{}{W}{\varss{x \mapsto N'_1}{\varss{\sigma}{M}}}$ in $\Lambda(P)$, such that
    $\Lambda' = Q'\ N_2'\ \ldots\ N_n'$. The crucial fact is that $Q'$ is
    actually strongly computable. Step by step :
    \begin{enumerate}
      \item By induction hypothesis, $\varss{x \mapsto N'_1}{\varss{\sigma}{M}}$ is
        strongly computable.
      \item By a general lemma, $M \in \mathbf{SC} \implies \refssd{}{V}{M} \in
        \mathbf{SC}$ for suitable $\multv{V}$. In particular, this means $Q' \in
        \mathbf{SC}$.
      \item Then, we can construct a substitution $\multv{X}$, such that
        starting from $\Lambda(Q,\multv{X},\seq{N}) = \refssd{}{X}{Q\ N_1\ \ldots\ N_n}$ we can mimic the
        reduction steps of $\Lambda(P)$ and get a $\Delta$ such that $\Lambda(Q) \to^\ast \Delta$
        with $\Lambda' \rincl \Delta$.
    \end{enumerate}
    $Q$ being strongly computable, $\Lambda(Q,\multv{X},\seq{N})$ (hence
    $\Delta$) is strongly normalizing, and we conclude once again with
    Lemma~\ref{lemma-ord-simulation}.\qed
\end{description}

The reader may found other cases of the proof sketched in \ref{ap:sound-adeq}.
Finally, together with Lemma \ref{lemma-soundness} (with $n=0$) and Lemma
\ref{lemma-adequacy} one can prove strong-normalization for {\lthis}. Moreover,
the reduction is locally confluent (see \ref{ap:confluence}): we deduce the
confluence of the language.

\begin{theorem}[Termination]\label{theorem-termination}
  All well-typed closed terms are strongly normalizing.\qed
\end{theorem}

\begin{corollary}[Confluence]\label{lemma-confluence}
  The reduction is confluent on typed terms.\qed
\end{corollary}

\section{Discussion}
\label{sec:discussion}

\subsection{Comparison with Other Languages}
One may wonder how {\lthis} compares to other concurrent calculi and especially
the language {\lamadio} presented in Section~\ref{sec:amadio}. In
particular, {\lthis}
is almost an explicit substitution version of {\lamadio}.
Indeed, it turns out that we can define a translation of {\lamadio}
to {\lthis}. The weak reduction of {\lthis} prevents variable substitutions from
percolating under abstractions, and translated terms may evaluate to closures as
$\lambda x.(\varss{\sigma}{M})$ instead of the expected $\lambda
x.(M\{x_1/\sigma(x_1),\ldots,x_n/\sigma(x_n)\})$ if
$\dom{\sigma}=\{x_1,\ldots,x_n\}$. Up to this difference (that can be properly
formalized -- see Appendix~\ref{ap:simul}) there is a simulation of {\lamadio} in {\lthis}.

More generally, we followed the design choice of adopting cumulative stores,
while many languages in the literature and in practice follow an erase-on-write
semantics. Remarkably, our choice makes the version with explicit substitutions
asynchronous, as various upward and downward substitutions may be reduced
arbitrarily without the need of any scheduling. Another point that
justifies its introduction is that such calculi simulate a lot of other
paradigms, such as erase-on-write or communication channels for example, as
mentioned in ~\cite{madet:tel-00794977}. This means that the termination of the cumulative
store version implies the termination of the aforementioned variants.
To illustrate our point, let us quickly sketch how a calculus with explicit
substitutions with an erase-on-write semantics could be devised. First, encode
$\set{r}{V}$ as $((\lambda x. \refrawsu{}{r \mapsto [V]}{\ast})\,\ast)$. Then,
when an upward substitution becomes reducible, apply all possible downward and
upward structural rules until it is not possible anymore. Finally, instead of
merging reference substitutions, the upper one erases the lower one. Its termination
follows immediately from the one of {\lthis}.

\subsection{Globality, Locality and Linear Logic}
Linear logic's proof nets are graphical representations of proofs as
graphs endowed with a local cut-elimination procedure. They are strongly connected to
systems with explicit substitutions (see e.g. \cite{Accattoli:2015:PNC:2852784.2853038}).
A lot of calculi have been encoded in proof nets or related systems:
call-by-value and call-by-name $\lambda$-calculi~\cite{MARAIST1995370},
$\pi$-calculus with limited replication~\cite{EHRHARD2010606},
$\lambda$-calculus with references~\cite{tranquilli:hal-00465793}, {\em etc}.
These representations naturally lead to parallel
implementations~\cite{mackie94,Pedicini:2007,Pinto:2001}, extend to richer
logics~\cite{ehrhard:hal-00150274} and form the basis for concrete operational
semantics in the form of token-based automata~\cite{DANOS199640}.
Our future goal is to push further this correspondence by modeling a language featuring
concurrency, references and replication. The constructs of {\lthis} are inspired
by the approach of~\cite{EHRHARD2010606} and~\cite{tranquilli:hal-00465793}.
{\lthis} can be seen as a calculus-side version of some kind of proof nets. The
translation and simulation of {\lamadio} in {\lthis} could be described as a
compilation from a {\em global shared memory} model to a {\em local message
passing} one, in line with proof nets' philosophy. The correctness of this
compilation requires that a well-typed strongly normalizing term in the initial
language is also strongly normalizing in the target language, and this is what
this paper achieves.

\section{Conclusion}
\label{sec:ccl}

In this paper, we presented a lambda-calculus with concurrence and
references, featuring explicit substitutions for both variables and
references. We discussed the issues explicit substitutions raise with
respect to termination and explained how standard techniques fail to
address them.

The main contribution of the paper is a solution to this problem.
Reminiscent of Game Semantics, the proof technique we apply is
interesting in its own right. Based on an interactive point of view,
it is reasonable to expect that the general methodology we present
can be extended to other settings, such as proof nets or concurrent
calculi.

Finally, with this work we open the way to an embedding of a calculus with references into
differential proof nets, which has been one of our leading motivation
for this work with the hope that these results may be as fruitful as
they have been in the study of lambda-calculus.

\bibliographystyle{eptcs}
\bibliography{info}

\newpage
\appendix

\section{Weak confluence}
\label{ap:confluence}

\begin{lemma}{Critical pairs}\label{lemma-critical-pairs}\\
  We write $E_1 \# E_2$ for two contexts $E_1,E_2$ if $E_1 \neq E_2$, and each one
  is not a prefix of the other, ie $\forall E, E_1 \neq E_2[E]$ and $E_2 \neq
  E_1[E]$. In the following, we write the reduction rules as
  $S[N] \to S[N'] + \mathbf{M'}$, where $\mathbf{M}'$ is equals to $\mathbf{0}$ unless when
  \rrulet{subst-r}{get} occurs where it may have additionnal terms.
  \paragraph{}
  If $\mathbf{M} \to \mathbf{M_1}$ and $\mathbf{M} \to \mathbf{M_2}$ with
  $\mathbf{M_1} \neq \mathbf{M_2}$, then one of the assertion holds :
  \begin{enumerate}
    \item The two rules are of the form $S_i[N_i] \to S_i[N_i'] + \mathbf{M}'_i$
      with $S_1 \# S_2, \mathbf{M} = S_i[N_i] = \sum_i N_i$
    \item The two rules are of the form $S[C_i[N_i]] \to S_i[C_i[N_i']] +
      \mathbf{M}'_i$
      with $C_1 \# C_2, C_i[N_i] = \parallel_i N_i$
    \item The two rules are of the form $S[C[N_1]] \to S[C[N_1']]$ and
      $S[C[C_2[N_2]]] \to S_2[C[C_2[N_2']]] + \mathbf{M}'_2$,
      the first rule being \rrule{subst-r'}{\parallel}.
    \item The two rules are of the form $S[C[E[E_i[N_i]]] \to
      S_i[C[E[E_i[N_i']]]] + \mathbf{M}_i'$ with $E_1 = \refssl{r}{V}{E_1'[.]\ E_2'[N_2]}$ and $E_2
      = \refssl{r}{V}{E_1'[N_1]\ E_2'[.]}$, $N_i \to N_i'$
    \item The two rules are of the form $S[C[E[N_1]]] \to S[C[E[N_1']]$ and
      $S[C[E[E_2[N_2]]]] \to S_2[C[E[E_2[N_2']]]] + \mathbf{M}'_2$, with one
      of the following :
      \begin{enumerate}
        \item $N_1 = \refssd{r}{V}{M'}$ and the applied rule is
          \rrule{subst-r}{app}, \rrule{subst-r}{subst-r'} or
          \rrule{subst-r}{merge}
        \item $N_1 = \refssl{r}{V}{(\refssu{s}{U}{P})\
        E_2'[N_2]}$ or $N_1 = \refssl{r}{V}{E_2'[N_2]\
          (\refssu{s}{U}{P})}$
        \item  $N_1 = \refssl{r}{V}{(\refssu{s}{U}{E_2'[N_2]})\
          P}$ or $N_1 = \refssl{r}{V}{P\
          (\refssu{s}{U}{E_2'[N_2]})}$
        \item $C = E = [.]$, $N_1 = \refssu{r}{V}{E_2[N_2]}$ and the first rule
          used is \rrule{subst-r'}{\top}
      \end{enumerate}
  \end{enumerate}
\end{lemma}

\begin{proof}
We can write $\mathbf{M}$ in a unique way
(modulo structual rules) as a sum of parallel of simple terms :
$$\mathbf{M} = \sum_k M_k,\  M_k = \parallel_{i} M^i_k$$
\begin{itemize}
  \item The two reductions rules have a premise of the form $S[N_i]$.
    Identifying the terms of both sums, we can write $S_1 = [.] + \sum_{k \neq
    k_1} M_k$ and $S_2 = [.] + \sum_{k \neq k_2} M_k$. If $k_1 \neq k_2$ we are in
    the case (1), or $S_1=S_2$.

  \item If one of the rule used (let say the first one) is
    \rrule{subst-r'}{\top}, then $M_{k_1} =
    \refssu{r}{V}{M_{k_1}'}$. Since $\mathbf{M_1} \neq
    \mathbf{M_2}$, the second rule can't be the same and is of the form
    $E[T] \to E[T']$. This is the case (5d) of the lemma.

  \item Otherwise, the premises of the two rules have the form $S[C_1[E_1[P_1]]]
    \to S[C_1[E_1[P_1']]] + \mathbf{M'}_1$ and $S[C_2[E_2[P_2]]] \to
    S[C_2[E_2[P_2']]] + \mathbf{M'}_2$.  If $C_1 \# C_2$, we are in case (2). If
    not, then either $C_1=C_2$ or $C_2=C_1[C_2'],C_2' \neq [.]$ but the only
    rule that matches a parallel is \rrule{subst-r'}{\parallel}, and this is
    case (3). We assume from now on that $C_1=C_2$. We decompose $E_1$ and $E_2$
    by their greatest common prefix, such that $E_1=E[E_1']$ and $E_2=E[E_2']$
    with either $E_1' = E_2' = [.]$, or $E_1' = [.],\ E_2' \neq [.]$, or $E_1 \#
    E_2$. The former is excluded since the reducts $\mathbf{N_1}$ and
    $\mathbf{N_2}$ are assumed differents, and no rule have overlapping redex on
    base cases (when $C$ and $E$ are empty).  \\
  \item
    If $E_1' \# E_2'$, since $E_1'[P_1]
    = E_2'[P_2]$, then $E_1$ must be of the form
    $\refssl{r}{V}{E_1''\ R}$, $E_2 =
    \refssl{r}{V}{L\ E_2''}$ with $L = E_1''[P_1]$ and $R =
    E_2''[P_2]$. This is case (4).
  \item
    Assume now that one of the two (let say $E_1'$) is $[.]$.
    Then $P_1$ and $E_2'$ have a common prefix. If $P_1$ is an application,
    it can't be the premise of the $(\beta_v)$ rule with $P_1 =
    \refssl{r}{U}{(\lambda x . M)\ V}$, because then $E_2' =
    \refssl{r}{U}{E_2''\ V}$ or $E_2' =
    \refssl{r}{U}{(\lambda x . M)\ E_2''}$ but no non-empty context
    verifies $E_2''[P_2] = V$ for a value $V$. Thus it must be the premise of
    \rrulet{subst-r'}{app}, and this corresponds to cases (5b) and (5c).

  \item
    If $P_1$ is not an application, since it must be both the premise of a rule
    and prefix of the context $E_2'$, the only remaining possibility is $P_1 =
    \refssd{r}{V}{\widetilde{P}_1}$. Then $\widetilde{P}_1$ can't be a value
    $y,\ast,\lambda y. \_$ or $\get{r}, \_ \parallel \_,
    \refssu{s}{U}{\_}$, because these constructors can't be in $E_2'$ : we are in case (5a).
\end{itemize}
\end{proof}

\begin{lemma}{Weak confluence}\label{lemma-weak-confluence}\\
  Let $\mathbf{M}$ be a term such that $\mathbf{M} \to
  \mathbf{M_1}$ and $\mathbf{M} \to \mathbf{M_2}$. Then
  $$\exists \mathbf{M'}, \mathbf{M_1} \to^\ast \mathbf{M'} \text{ and }
  \mathbf{M_2} \to^\ast \mathbf{M'}$$
\end{lemma}

\begin{proof}
We can write $\mathbf{M},\mathbf{M_1}$ and $\mathbf{M_2}$ in a unique way
(modulo structual rules) as a sum of parallel of simple terms :
$$\mathbf{M} = \sum_k M_k,\ \mathbf{M_1} = \sum_k M^1_k,\
\mathbf{M_2} = \sum_k M^2_k$$
Let us process all the possible cases of \ref{lemma-critical-pairs}, assuming that $\mathbf{N_1} \neq \mathbf{N_2}$:
\begin{enumerate}
  \item We have $S_1 \# S_2$, by identifying each terms, $\exists k_1 \neq k_2, S_i =
    [.] + \sum_{k \neq k_i}$, and we have
    \begin{tabred}
      \mathbf{M}_1 & = & \sum_{k \neq k_1} M_k + N_1' + \mathbf{M}'_1 \\
        & \to & \sum_{k \neq k_1,k_2} M_k + N_1' + N_2' + \mathbf{M}'_1 +
          \mathbf{M}'_2
    \end{tabred}
    as well as $\mathbf{M_2}$.
  \item Let write $C_i[N_i] = \parallel_{l \in
      \mathcal{L}} P_l$. Then
    there exists $\mathcal{L}_1, \mathcal{L}_2 \subseteq \mathcal{L}$. Since
    $C_1 \# C_2$, we have $\mathcal{L}_1 \neq \mathcal{L}_2$ and $\mathcal{L}_1
    \not \subseteq \mathcal{L}_2$ and $\mathcal{L}_2 \not \subseteq
    \mathcal{L}_1$, such that $C_i = [.] \parallel (\parallel_{l
      \in \mathcal{L}_i} P_l)$ and $N_i = \parallel_{l \notin \mathcal{L}_i} P_l$.
    The only rule that has a parallel of terms as premise is
    \rrule{subst-r'}{\parallel}. Thus $\mathcal{L}_i$ are either singletons (if
    the corresponding rule is not \rrule{subst-r'}{\parallel}) or have size two. If
    they are disjoint, then $\mathcal{L}_2 \subseteq \mathcal{L} \setminus
    \mathcal{L}_1$ and :
    \begin{tabred}
      \mathbf{M}_1 & = & S[N_1' \parallel (\parallel_{l \notin \mathcal{L}_1} P_l)] + \mathbf{M}'_1 \\
        & = &
          S[N_1' \parallel N_2 \parallel (\parallel_{l \notin \mathcal{L}_1 \cup
            \mathcal{L}_2} P_l)] + \mathbf{M}'_1 \\
        & \to &
          S[N_1' \parallel N_2' \parallel (\parallel_{l \notin \mathcal{L}_1 \cup
          \mathcal{L}_2} P_l)] +
          \mathbf{M}'_1 + \mathbf{M}'_2 \\
    \end{tabred}
    as do $\mathbf{M_2}$.\\
    The only remaining case is if both rules are \rrule{subst-r'}{\parallel} and
    $\mathcal{L}_1 \cap \mathcal{L}_2 = \{ l_0 \}$. We write $\mathcal{L}_1 = \{
    l_1, l_0 \}, \mathcal{L}_2 = \{ l_2, l_0 \}$ and $\mathcal{L_3} =
    \mathcal{L} \setminus \{ l_0, l_1, l_2 \}$. If $P_{l_0} =
    \refssu{r}{V}{P'}$ is the "active" upward substitution in both reduction, we
    have
    \begin{tabred}
      C_1[N_1'] & = & \refssu{r}{V}{(\refssd{r}{V}{P_{l_1}} \parallel P')}
        \parallel P_{l_2} \parallel (\parallel_{l \neq \mathcal{L}_3} P_l) \\
        & \to & \refssu{r}{V}{(\refssd{r}{V}{P_{l_1}} \parallel P' \parallel
          \refssd{r}{V}{P_{l_2}})} \parallel (\parallel_{l \neq \mathcal{L}_3}
          P_l) \\
    \end{tabred}
    If $P_{l_0}$ is the "passive" term in both reductions, with $P_{l_1} =
    \refssu{r}{V}{P_{l_1}'}$ and  $P_{l_2} = \refssu{s}{U}{P_{l_2}'}$, then
    \begin{tabred}
      C_1[N_1'] & = & \refssu{r}{V}{(P'_{l_1} \parallel \refssd{r}{V}{P_{l_0}})}
      \parallel \refssu{s}{U}{P'_{l_2}} \parallel (\parallel_{l \neq \mathcal{L}_3} P_l) \\
      & \to^\ast & P_{l_1} \parallel
        (\refssd{r}{V}{P_{l_0}}) \parallel
        (\refssd{r}{V}{\refssu{s}{U}{P'_{l_2}}}) \\
      & &
        \parallel (\parallel_{l \neq \mathcal{L}_3}
        \refssd{r}{V}{P_l}) \\
      & \to & P_{l_1} \parallel
        (\refssd{r}{V}{P_{l_0}}) \parallel
        (\refssu{s}{U}{\refssd{r}{V}{P'_{l_2}}}) \\
      & &
        \parallel (\parallel_{l \neq \mathcal{L}_3}
        \refssd{r}{V}{P_l}) \\
        & \to^\ast & (\refssd{s}{U}{P_{l_1}}) \parallel
        (\refssd{t}{W}{P_{l_0}}) \parallel
        (\refssd{r}{V}{P'_{l_2}}) \\
      & &
        \parallel (\parallel_{l \neq \mathcal{L}_3}
        \refssd{t}{W}{P_l}) \\
    \end{tabred}
    using repeated \rrule{subst-r'}{\parallel}, \rrule{subst-r'}{\top},
    \rrulet{subst-r}{subst-r'} and
    \rrule{subst-r}{merge}.
    \\
    Finally, if $P_{l_0}$ is active in of the two (let say the first) and
    passive in the other, meaning that $P_{l_0} =
    \refssu{r}{V}{P_{l_0}'},P_{l_2} = \refssu{s}{U}{P_{l_2}'}$, then

    \begin{tabred}
      C_1[N_1'] & = & \refssu{r}{V}{(\refssd{r}{V}{P_{l_1}} \parallel P_{l_0}')}
      \parallel (\refssu{s}{U}{P_{l_2}}) \parallel (\parallel_{l \neq \mathcal{L}_3} P_l) \\
        & \to & \refssu{r}{V}{(\refssd{r}{V}{P_{l_1}} \parallel P' \parallel
          \refssd{r}{V}{P_{l_2}})} \parallel (\parallel_{l \neq \mathcal{L}_3}
          P_l) \\
    \end{tabred}
    \begin{tabred}
      C_1[N_1'] & = & \refssu{r}{V}{(\refssd{r}{V}{P_{l_1}} \parallel P'_{l_0})}
      \parallel \refssu{s}{U}{P'_{l_2}} \parallel (\parallel_{l \neq \mathcal{L}_3} P_l) \\
      & \to^\ast & (\refssd{r}{V}{P_{l_1}}) \parallel
        P'_{l_0} \parallel
        (\refssd{r}{V}{\refssu{s}{U}{P'_{l_2}}}) \\
      & &
        \parallel (\parallel_{l \neq \mathcal{L}_3}
        \refssd{r}{V}{P_l}) \\
      & \to &  (\refssd{r}{V}{P_{l_1}}) \parallel
        P'_{l_0} \parallel
        (\refssu{s}{U}{\refssd{r}{V}{P_{l_2}}}) \\
      & &
        \parallel (\parallel_{l \neq \mathcal{L}_3}
        \refssd{r}{V}{P_l}) \\
        & \to^\ast & (\refssd{t}{W}{P_{l_1}}) \parallel
        (\refssd{s}{U}{P_{l_0}}) \parallel
        (\refssd{r}{V}{P_{l_2}}) \\
      & &
        \parallel (\parallel_{l \neq \mathcal{L}_3}
        \refssd{t}{W}{P_l}) \\
    \end{tabred}
    On the other side,
    \begin{tabred}
      C_2[N_2'] & = & P_{l_1} \parallel
        \refssu{s}{U}{( (\refssd{s}{U}{\refssu{r}{V}{P'_{l_0}}}) \parallel
        P'_{l_2})}
        \parallel (\parallel_{l \neq \mathcal{L}_3} P_l) \\
      & \to^\ast & (\refssd{s}{U}{P_{l_1}}) \parallel
        (\refssd{s}{U}{\refssu{r}{V}{P'_{l_0}}}) \parallel P'_{l_2} \\
      & &
          \parallel (\parallel_{l \neq \mathcal{L}_3}
          \refssd{s}{U}{P_l}) \\
      & \to & (\refssd{s}{U}{P_{l_1}}) \parallel
        (\refssu{r}{V}{\refssd{s}{U}{P'_{l_0}}}) \parallel P'_{l_2} \\
      & &
          \parallel (\parallel_{l \neq \mathcal{L}_3}
          \refssd{s}{U}{P_l}) \\
      & \to^\ast & (\refssd{t}{W}{P_{l_1}})
        \parallel (\refssd{s}{U}{P'_{l_0}})
        \parallel (\refssd{r}{V}{P'_{l_2}}) \\
      & &
          \parallel (\parallel_{l \neq \mathcal{L}_3}
          \refssd{t}{W}{P_l}) \\
    \end{tabred}

  \item
    $C_2 = \refssu{r}{V}{Q} \parallel C_2'$. Let write $C = \parallel_{l} P_l$
    and $C_2 = [.] \parallel P'$.
    then $C[N_1'] = \refssu{r}{V}{(Q \parallel \refssd{r}{V}{(N_2 \parallel P')})}$.
    \begin{itemize}
      \item Either $N_2 = \refssu{s}{U}{N_2'} \parallel Q'$ and
        \begin{tabred}
          C[N_1'] & \to^\ast & Q \parallel
            \refssd{r}{V}{((\refssu{s}{U}{N_2'}) \parallel Q')}
            \parallel (\parallel_{l} \refssd{r}{V}{P_l}) \\
          & \to^\ast & Q \parallel
            (\refssu{s}{U}{\refssd{r}{V}{N_2'}}) \parallel (\refssd{r}{V}{Q'})
            \parallel (\parallel_{l} \refssd{r}{V}{P_l}) \\
          & \to^\ast & (\refssd{s}{U}{Q})
          \parallel (\refssd{r}{V}{N_2'})
          \parallel (\refssd{t}{W}{Q'})
          \parallel (\parallel_{l} \refssd{t}{W}{P_l}) \\
        \end{tabred}
        and
        \begin{tabred}
          C[C_2[N_2']] & = & (\refssu{r}{V}{Q})
            \parallel \refssu{s}{U}{(N_2' \parallel \refssd{s}{U}{Q'})}
            \parallel (\parallel_{l} P_l) \\
          & \to^\ast &
            (\refssd{s}{U}{\refssu{r}{V}{Q}}) \parallel
            N_2' \parallel (\refssd{s}{U}{Q'})
            \parallel (\parallel_{l} \refssd{s}{U}{P_l}) \\
          & \to^\ast &
            (\refssd{s}{U}{Q}) \parallel
            (\refssd{r}{V}{N_2'}) \\
          & &
            \parallel (\refssd{t}{W}{Q'})
            \parallel (\parallel_l \refssd{t}{W}{P_l}) \\
        \end{tabred}
      \item Otherwise, $N_2$ is a premise of the form $E[Q']$ and
        $S[C[C_2[E[Q']]] \to S[C[C_2[E[Q'']]]] + \mathbf{M}_2'$. Then
        \begin{tabred}
          \mathbf{M_1} & \to^\ast & S[Q \parallel
            \refssd{r}{V}{E[Q']}
            \parallel (\parallel_{l} \refssd{r}{V}{P_l})] \\
          & \to & S[Q \parallel
            \refssd{r}{V}{E[Q']}
            \parallel (\parallel_{l} \refssd{r}{V}{P_l})] + \mathbf{M_2}' \\
        \end{tabred}
        On the other side,
        \begin{tabred}
          \mathbf{M_2} & = & S[(\refssu{r}{V}{Q}) \parallel
            E[Q''] \parallel
            (\parallel_{l} P_l)] + \mathbf{M_2}' \\
          & \to^\ast & S[Q \parallel
            (\refssd{r}{V}{E[Q'']}) \parallel
            (\parallel_{l} \refssd{r}{V}{P_l})] + \mathbf{M_2}' \\
        \end{tabred}
    \end{itemize}

  \item
    \begin{tabred}
      S[C[E[E_1[N_1']]]] + \mathbf{M}'_1 & = &
        S[C[E[\refssl{r}{V}{E_1'[N_1']\ E_2'[N_2]}]]] + \mathbf{M}'_1 \\
      & \to &
        S[C[E[\refssl{r}{V}{E_1'[N_1']\ E_2'[N_2']}]]] + \mathbf{M}'_1 +
        \mathbf{M}'_2 \\
    \end{tabred}

  \item
    \begin{enumerate}
      \item The applied rule is either :
        \begin{itemize}
          \item \rrulet{subst-r}{app} and $E_2 = \refssl{s}{U}{E_2'\ Q}$ or $E_2
            = \refssl{s}{U}{Q\ E_2'}$
          \item \rrulet{subst-r}{subst-r'} and $E_2 = \refssu{s}{U}{E_2'}$
          \item \rrulet{subst-r}{merge} and $E_2 = \refssd{s}{V}{E_2'}$
        \end{itemize}
        In the three cases, it is clear that the reductions are independant :
        the first one can be performed in $\mathbf{M_2}$ and vice-versa to get a
        common reduct.
      \item
      \item
      \item \ldots The same argument applies to the other four cases.
    \end{enumerate}
\end{enumerate}
\end{proof}

\section{Termination}
\label{ap:termination}

\subsection{Preorder on terms}
\label{ap:preorder}

%
\begin{definition}[Reachability and Associated Preorder]\label{def:reach-preorder}%
  \rm
  Let $M$ be a term, and $N$ an occurrence of a subterm in $M$
  that is not under an abstraction. We define
  $\reach{N}{M}$, a reference substitution, as the merge of all substitutions that are in scope
  of this subterm in $M$, as follows. Recall Notation~\ref{not:subst}.
  \begin{itemize}
    \item If $M=N$ then $\reach{N}{M}$ is nowhere defined.
    \item If $M=\refssd{s}{U}{M'}$ then $\reach{N}{M} =
      \multv{U},\reach{N}{M'}$, the juxtaposition of $\multv{U}$ and $\reach{N}{M'}$.
    \item If $M=\varss{\sigma}{M'}$ then $\reach{N}{M} = \metas{\sigma}{\reach{N}{M'}}$.
    \item If $M = \refssu{s}{U}{M'}$ then $\reach{N}{M} =
      \reach{N}{M'}$.
    \item If $M=\refssl{s}{U}{M_1\ M_2}$ or $M = M_1 \parallel M_2$, let $i$ be the index such that
      $N$ occurs in $M_i$, then $\reach{N}{M} = \reach{N}{M_i}$.
  \end{itemize}
  We define the skeleton of a term $\skel{M}$ by removing all downward
  reference substitutions that are not under an abstraction.
\end{definition}

\begin{definition}[Preorder]
  We say that $M
  \rincl N$
  if:
  \begin{itemize}
    \item $\skel{M}$ = $\skel{N}$, and thus we can put in a one-to-one correspondence the occurrences of
      $\get{r}$ and $\refssl{}{V}{(M_1\,M_2)}$ subterms of $M$ and $N$
    \item For all such $\get{r}$ occurrences, $\reach{\get{r}}{M} \subseteq
      \reach{\get{r}}{N}$
    \item For all such $M' = \refssl{r}{V}{(M_1\,M_2)}$ corresponding to
      $N' = \refssl{s}{U}{(N_1\,N_2)}$, then $\reach{M'}{M},\multv{V}
      \subseteq \reach{N'}{N},\multv{U}$
  \end{itemize}
  Similarly, we say that $M \rdiff{V} N$ if the
  difference between reachability sets involved in the definition is somehow
  ``bounded'' by $\multv{V}$:
  \begin{itemize}
    \item $\skel{M}$ = $\skel{N}$
    \item For all such occurrences of $\get{r}$, $\reach{\get{r}}{M} \subseteq
      \reach{\get{r}}{N} \subseteq \reach{\get{r}}{M},\multv{V}$
    \item For all such $M' = \refssl{s}{U}{M_1\ M_2}$ corresponding to
      $N' = \refssl{t}{W}{N_1\ N_2}$, then $\reach{M'}{M},\multv{W} \subseteq
      \reach{N'}{N},\multv{U} \subseteq \reach{M'}{M},\multv{W},\multv{V} $
  \end{itemize}
  The relations $\rincl$ and $\rdiff{V}$ are partial preorders on
  terms. 
\end{definition}

$M \rincl M'$ if $M$ and $M'$ have the same structure but the
available substitutions in scope of each $\get{r}$ in $M$ are contained in
$M'$ ones. Thus, $M'$ can do at least everything $M$ can do. The second preorder
$\rdiff{V}$ controls precisely what the difference between reachability sets
can be. The following properties make these intuitions formal:

\begin{lemma}{Invariance by \rrulet{subst-r}{}
reductions}\label{lemma-ord-invariance}\\
Let $M \rincl N$ (resp. $M \rdiff{V} N$).
  \begin{itemize}
    \item If $M \to M'$ by a \rrulet{subst-r}{} rule except
      \rrulet{subst-r}{get} then $M' \rincl N$ (resp. $M' \rdiff{V} N$)
    \item If $N \to N'$ by a
      \rrulet{subst-r}{} by a \rrulet{subst-r}{} rule except
      \rrulet{subst-r}{get} then $M \rincl N'$ (resp. $M \rdiff{V} N'$).
  \end{itemize}
\end{lemma}

\begin{proof}
 Clearly, a \rrulet{subst-r}{} rule does not modify the skeleton, so $\skel{M} =
 \skel{M'} = \skel{N}$. It is also almost immediate to see that rules that
 propagate reference substitutions or \rrulet{subst-r}{val} that erases the ones
 only in scope of a value do not modify $\reach{\get{r}}{M}$ for an
 occurence $\get{r}$ in $M$, nor do they modify $\reach{M'}{M} + \multv{V}$ for
 $M' = \refssl{r}{V}{M_1\ M_2}$.
\end{proof}

\begin{lemma}{Simulation}\label{lemma-ord-simulation}\\
  Let $M \rincl N$ (resp. $M \rdiff{V} N$).
  \begin{enumerate}
    \item If $M \to M'$ then $\exists n \geq 0,\ N \to^n N'$ such that $M'
      \rincl N'$ (resp. $M' \rdiff{V} N'$). If the
      applied rule is not a \rrulet{subst-r}{} or is \rrulet{subst-r}{get}, then
      $n > 0$.
    \item Corollary. If $M \rincl N$, then if $N$ is strongly normalizing, so is
      $M$.
    \item Corollary. If $M \req N$, then $M$ is strongly normalizing iff $N$
      is.
  \end{enumerate}
\end{lemma}

\begin{proof}
  We will first prove that if $M = C[E[P]$, then $N \to^\ast C[E'[P']]$. One may
  just have to apply the rule \rrule{subst-r}{\parallel} until it
  is not possible anymore to get $N \to C[N']$. Then, $E[P]$ and $N'$ having the
  same skeleton,
  $N'$ can be written as $E'[P']$ where $E'$ is $E$ with
  additionnal downward references substitutions, and $P$ and $P'$ have the same
  skeleton and the same head constructor (if $P'$ have additionnal substitutions in
  head position one can always include it in $E'$ : we actually take the maximal $E'$
  that satisfies the decomposition).
  The reachability sets of subterms in $C$ (respectively $E$,$E'$) only depends
  on $C$ (resp. $C,E$ and $C,E'$). The reachability sets of subterms in $P$
  (resp. $P'$) are unions of substitutions occuring in $P,E$ (resp.
  $P',E'$) and $C$.

  \begin{itemize}
    \item If the rule is one of the \rrulet{subst-r}{} except
      \rrulet{subst-r}{get}, by \ref{lemma-ord-invariance}, $n=0$ works.
    \item \rrulet{subst-r}{get} : $P = \refssd{r}{V}{\get{r}} \to V \in
      \multv{V}(r)$ and $P =
      \refssd{r'}{V'}{\get{r}}$. $\multv{V}(r)$ is in $\reach{\get{r}}{N}$ so by iterated application of \rrulet{subst-r}{}
      rules except \rrulet{subst-r}{get} and \rrulet{subst-r}{val}, we can push
      (without modifying the skeleton nor the reachability sets) the
      corresponding substitutions down to $\get{r}$ in $P'$ and we can do the
      same reduction. All the other reachability sets of gets or application or
      left unmodified.
    \item $(\beta_v)$ : $P = \refssl{r}{V}{(\lambda x. Q)\ V} \to
      T = \refssd{r}{V}{\vars{x}{V}{Q}}$. Up to
      \rrulet{subst-r}{val} reductions, $P' = \refssl{r'}{V'}{(\lambda x. Q)\ V} \to
        T' = \refssd{r'}{V'}{\vars{x}{V}{Q}}$. The condition on reachabiliy sets for
        application in the defintion of $\rincl$ precisely ensures that all the
        gets and applications in $Q$ have the same reachability in
        $C[E[T]]$ and in $C[E'[T']]$.
  \end{itemize}

  \rrulet{subst}{} rules
  \begin{itemize}
    \item \rrulet{subst}{var} : $P = \varss{\sigma}{x} \to x$ or $V(x)$, and $P' =
      \varss{\sigma}{\refssd{r}{V}{x}} \to \varss{\sigma}{x} \to x$ or $V(x)$.
    \item \rrulet{subst}{app} : $P = \varss{\sigma}{\refssl{r}{V}{Q_1\ Q_2}} \to
      T = \refrawsl{\vect{r}}{(\varss{\sigma}{V})}{(\varss{\sigma}{Q_1})\
      (\varss{\sigma}{Q_2})}$. We have $P' =
      \varss{\sigma}{\refssl{r'}{V'}{(\refssl{s}{U}{Q'_1})\
      (\refssl{t}{W}{Q'_2})}} \to^\ast
      \refrawsl{r'}{\varss{\sigma}{V'}}{(\refrawsd{s}{\varss{\sigma}{U}}{\varss{\sigma}{Q_1}})\
      (\refrawsd{s}{\varss{t}{V}{W}}{\varss{\sigma}{Q_2}})}$. By definition of
      reachability sets, they are invariant by all the rule applied.
    \item We proceed the same way for other cases : the var substitution just
      go through the additionnal references substitutions, and by design,
      reachability sets are not modified.
  \end{itemize}

  \rrulet{subst-r'}{} rules
  \begin{itemize}
    \item Upward substitutions commute with downward ones without interacting.
      On the other hand, they can span new downward substitutions but in this
      case, they do it in the same way for both $P$ and $P'$ and thus do not
      modify the inclusion relation on reachability sets.
  \end{itemize}
\end{proof}

\subsection{Soundness and Adequacy}
\label{ap:sound-adeq}

\begin{lemma}{Characterization}\label{lemma-charact}\\
  Let
  \begin{itemize}
    \item $\alpha = A_1 \sur{\to}{e_1} \ldots \sur{\to}{e_{n-1}} A_n
      \sur{\to}{e_n} \beta$ with $\beta = \text{Unit} \mid \mathbf{B}$
    \item $\vdash M : (\alpha,e)$
    \item $N_i \in \mathbf{SC}(A_i,e_i')$
    \item $\multv{U} \subseteq \mathbf{SC}$
  \end{itemize}
  with $\ptentre{i}{1}{n}, e'_i \subseteq e_i \subseteq e$ and $\dom{\multv{U}} \subseteq e$.
  We define
  \[
    \Lambda(M,\multv{U},\seq{N}) = \refssd{r}{\mathcal{U}}{M\ N_1\ \ldots\ N_n}
  \]
    Then $M \in \mathbf{SC}(\alpha,e)$ if and only if
  $\Lambda(M,\multv{U},\seq{N})$ is {\condSN} and {\condWB} for all $\multv{U},\seq{N}$
  satisfying the above conditions.
  In the following, we may conveniently omit some of the parameters
  $(M,\multv{U},\seq{N})$ of $\Lambda$.
\end{lemma}

\begin{proof}
  By induction on types.
\end{proof}

\begin{lemma}{Auxiliary results for soundness}\label{lemma-soundness-aux}\\
  Let $M,V \in \mathbf{SC},\ \multv{V} \subseteq \mathbf{SC}$ then
  \begin{enumerate}
    \item For any infinite reduction of $\Lambda(N)$, $N$ must be reduced at some
      point.
    \item If $\varss{\sigma}{N} \to M$ then $\varss{\sigma}{N} \in \mathbf{SC}$
    \item $\refssd{r}{V}{M} \in \mathbf{SC}$
    \item $\refssu{r}{V}{M} \in \mathbf{SC}$
    \item If $M \to M'$, then $M' \in \mathbf{SC}$
    \item $\varsmetas{x}{V}{\multv{V}} \subseteq \mathbf{SC}$
  \end{enumerate}
\end{lemma}

\begin{proof}
  \begin{enumerate}
    \item If the subterm $N$ is never reduced, the reducts of $\Lambda(N)$ are of
      the form $N\ N_1' \ldots N_n'$ (with some additionnal reference
      substitutions not written for conciseness) where each $N_i'$ can be seen as the result
      of an $(N_i,\seq{\multv{W}})$ reduction for some $\seq{\multv{W}}$ corresponding
      to the upward substitutions generated by the interaction with other
      subterms $N_1,\ldots,N_n$. All these subterms are {\condWB} and
      generate $\mathbf{SC}$ upward substitutions. $N_i$ being well behaved, there is a number of
      steps $n_i$ after wich the reduct $N'_i$ doesn't generate upward
      substitutions anymore. After $n= \vee n_i$ steps (actually one may have to take a
      bigger $n$ for the substitutions have to dispatch, but there is such a $n$), $N_i$ has reduced to some $P_i$
      and doesn't receive any downard substitution. If we gather all the downward substitutions
      delivered to $N_i$ (or, viewed as an $(N_i,\seq{\multv{W}})$ reduction, all
      the substitutions produced by $\to_i$ rules) as $\multv{X}$, then
      $N_i \rdiff{U} Q = \refssd{}{U}{N_i}$. Consider the first step of the
      $(N_i,\seq{\multv{W}})$ reduction $N_i \to N_i'$, either this is a usual step and then by
      \ref{lemma-ord-simulation} $Q \to^\ast Q'$ with $N_i'
      \rdiff{U} Q'$ or this is a $\to_j$ step and $N_i' = \refssd{}{Z}{N_i}$
      for some $\multv{Z} \subseteq \multv{U}$ and $N_i'
      \rdiffraw{\multv{U} \setminus \multv{Z}} Q$. By induction over the n firt
      steps, we get that $Q \to^\ast Q_0$ such that $P_i
      \req Q_0$. Since $N_i \in \mathbf{SC}$, $Q$ and all its reducts are
      {\condSN}, and so is $P_i$ thus the reduction of $\Lambda(N)$ must be finite.

    \item Let $M' = \varss{\sigma}{M}$, consider an infinite reduction of
      $\Lambda(M')$.  By the previous point, the subterm $M'$ must be reduced,
      and $M$ is the only possible reduct. But the reduct of $\Lambda(M')$
      we get is reachable from $\Lambda(M)$ which is {\condSN} thus the reduction
      must be finite and $\Lambda(M')$ is {\condSN}.

      Similarly, consider (for suitable $\seq{\multv{V}}$) a $(M',\seq{\multv{V}})$
      reduction. If $M'$ is never reduced, then no upward
      substitution is ever produced. If $M'$ is reduced at
      some point, then it is reduced to $M$ and thus produces a finite amount of
      $\mathbf{SC}$ upward substitutions since $M$ is {\condWB}. Hence $M'$ is {\condWB}.

    \item Consider $\Lambda(\refssd{r}{V}{M},\multv{U},\seq{N})$. Then
      $\Lambda(\refssd{r}{V}{M},\multv{U},\seq{N}) \rincl \Lambda(M,\multv{U} +
      \multv{V},\seq{N})$ but the latter is {\condSN}, since $M$ is $\mathbf{SC}$. So is
      the former by \ref{lemma-ord-simulation}.
      For {\condWB}, we can easily map a $(M',\seq{\multv{W'}})$ reduction to a
      $(M,\seq{\multv{W}})$ by just appending $\multv{V}$ to
      $\seq{\multv{W}}$ and
      start with a $\to_i$ reduction. Since $M \in \mathbf{SC}$, $M'$ is {\condWB}.

    \item Let $M' = \refssu{r}{V}{M}$. We proceed by induction on types. For
      base types, it is clear that $M'$ is {\condWB} iff $M$ is, and
      $\refssd{s}{U}{M'}$ has exactly the same reductions as $\refssd{s}{U}{M}$
      except for commutation of upward and downard substitutions, and a possible
      \rrule{subst}{\top}. Thus they are both {\condSN}.
      Now, for $\alpha = A \sur{\to}{e_1} \alpha'$, consider an infinite reduction of
      $\Lambda(M',\multv{U},\seq{N})$. If the upward substitution is never
      reduced, we can map this to an infinite reduction of
      $\Lambda(M,\multv{U},\seq{N})$ for the same reasons as above, but the
      latter is {\condSN}. Hence at some point the upward substitution must move up,
      so that the head term becomes $\refssu{r}{V}{\refssl{t}{W}{M''\
        (\refssd{r}{V}{N_1'}})}$ where $M''$ and $N_1'$ are reducts descendant
      from respectively $M$ and $N_1$. Using induction and points 5 and 3 of
      this lemma, $P = \refssl{t}{W}{M''\ (\refssd{r}{V}{N_1'}})$ is $\mathbf{SC}$. By
      induction, so is the new head redex, and the whole redex of $\Lambda$ is
      reachable from $\Lambda(\refssl{t}{W}{M\ N_1},\multv{U},\seq{N})$.

    \item This is straightforward from the definition of $\mathbf{SC}$ sets.
  \item Combining 2. and 4., using the fact that $\varsmetas{x}{V}{U}$ is the
      normal form of $\varss{\sigma}{U}$ using only \rrule{subst}{} rules, we get
      the result by induction on the length of the reduction.
  \end{enumerate}
\end{proof}

We are now able to sketch the proof of soundness and adequacy:

\begin{proof}\ref{lemma-soundness}
  We perform the proof by induction.
  \begin{itemize}
    \item $M=x$ : $\varss{\sigma}{x}$ reduces to $\sigma(x) \in \mathbf{SC}$, and
      we apply \ref{lemma-soundness-aux}
    \item $M=\ast$ : $\varss{\sigma}{\ast} \to \ast$.
    \item $M=\lambda y.M'$ : $\varss{\sigma}{M} \to \lambda x .
      (\varss{\sigma}{M'})$. Consider an infinite reduction $\Lambda(\lambda x.
      (\varss{\sigma}{M'}))$. By \ref{lemma-soundness-aux} 1), at some point a
      $(\beta_v)$ must occur, replacing the head redex by $P =
      \refssd{s}{U}{\vars{y}{X}{\varss{\sigma}{M'}}}$ which is $\mathbf{SC}$ by
      induction \ref{lemma-soundness-aux} 3). But this term can be reached from
      $\Lambda(P,\multv{U},\refssd{}{V}{N_{i+1}})$ (for some
      $\multv{V}$ that are $\mathbf{SC}$, the ones emitted by the reduction of
      $N_1$) which is {\condSN}. The same kind of
      argument show that $\Lambda(M,\multv{U},\seq{N})$ is {\condWB} : if the redex
      is not reduced, then all upgoing substitutions that come from $N_i$ must stop
      after a finite number of reductions and all contain $\mathbf{SC}$ terms,
      or the redex is reduced after a finite number of steps and from this point
      the term is a reduct of a {\condWB} one thus must be {\condWB} as well.

    \item $M = \get{r}$ : $\varss{\sigma}{\get{r}} \to \get{r}$. Let see that
      $\get{r} \in \mathbf{SC}$. Consider an infinite reduction of
      $\Lambda(\get{r})$. By \ref{lemma-soundness-aux} 1), the $\get{r}$ must be
      reduced at some point and from
      this point it is either replaced by a value from the substitution prefix,
      or by a value emitted by one of the $N_i$, all of these being
      $\mathbf{SC}$.
    \item $M = \varss{\tau}{M'}$, then $\varss{\sigma}{M} \to
      \varss{\mu}{M'}$. By \ref{lemma-soundness-aux},
      $\varsmetas{x}{V}{U} \in \mathbf{SC}$. $M'$ has a typing judgement of the form $x_1:A_1,\ldots,x_n : A_n, y_1 :
      B_1, \ldots, y_m : B_m \vdash M' : (\alpha,e)$. By induction,
      $\varss{\mu}{M'} \in
      \mathbf{SC}$.
    \item $M = \refssd{}{U}{M'}$ : $\varss{\sigma}{M} \to
      \refssd{}{\varsmetas{x}{V}{\multv{U}}}{\varss{\sigma}{M'}}$. By
      induction, $\varss{\sigma}{M'}$ is $\mathbf{SC}$ and by \ref{lemma-soundness-aux}
      so is $M$.
    \item $M = \refssu{r}{\mathcal{V}}{M'}$ : induction +
      \ref{lemma-soundness-aux}.
    \item $M = \refrawsl{(r_i)}{(\mathcal{V}_i)}{M_1\ M_2}$ : $\varss{\sigma}{M}
      \to M' = \refrawsl{r}{(\varss{\sigma}{\multv{V}})}{(\varss{\sigma}{M_1})\
      (\varss{\sigma}{M_2})}$. By induction and \ref{lemma-soundness-aux}, $\varss{\sigma}{M_i} \in \mathbf{SC}$
      and $\varss{x}{\multv{V}}$ are in $\mathbf{SC}$. Then by definition
      $M' \in \mathbf{SC}$.
    \item $M = M_1 \parallel M_2$ : $\vars{x}{V}{M} \to M'= M_1' \parallel M_2'$
      with $M_i' = \vars{x}{V}{M_i}$  and according to \ref{lemma-soundness-aux} it
      is sufficient to prove $M' \in \mathbf{SC}$. By induction, $M_i' \in
      \mathbf{SC}$. Let us take an $(M',r,\multv{V})$ reduction
      starting from $M'$. It can be associated to an $(M_i',r_i,\multv{V}_i)$
      reduction from the point of view of $M_i'$, where
      $r_i,\multv{V}_i$ are coming either from $(r,\multv{V})$ or from the
      upward substitution of $M'_{1-i}$, that are all in $\mathbf{SC}$. Since
      the $M_i$ are {\condWB}, so is $M'$.
      Now, the proof is very similar to the first point of
      \ref{lemma-soundness-aux} : Take an $M'$ reduction, since $M'_i$ are {\condWB}, after a finite
      number of steps, no more substitutions are exchanged. We can then smash
      all the substitutions received by $M'_i$ into a big one $(x,\multv{X})$
      and since $\refssd{x}{X}{M'_i}$ is {\condSN}, then the reduction must be
      finite.
  \end{itemize}

\end{proof}

\begin{proof}\ref{lemma-adequacy}
  We will prove additionnally by induction on types that $\mathbf{SC}(\alpha)
  \neq \emptyset$ for any $\alpha$.
  \begin{itemize}
    \item For $\alpha = \texttt{Unit} \mid \mathbf{B}$ it is immediate : $\ast
      \in \mathbf{SC}(\texttt{Unit})$ and $\ast \parallel \ast \in
      \mathbf{SC}(\mathbf{B})$. If $M \in \mathbf{SC}(\alpha)$
      had an infinite reduction, then so would $\refsrawsd{r}{V}{M}$
    \item For $\alpha = A \to \alpha'$, by induction there exists $M_{\alpha'}
      \in \mathbf{SC}(\alpha')$, then $\lambda x . M_{\alpha'}$ for $x$ not free in
      $M_{\alpha'}$ is in $\mathbf{SC}(\alpha)$. If $M \in \mathbf{SC}(\alpha)$ did not terminate, so would $M\
      N$ for $N \in \mathbf{SC}(A)$ which exists since the latter set
      is not empty.
  \end{itemize}
\end{proof}

\section{Relation between {\lthis} and {\lamadio}}
\label{ap:simul}
In the following, {\lamadio} stands for the version of the concurrent
$\lambda$-calculus described in the second chapter of
~\cite{madet:tel-00794977}.
\newcommand{\vstore}{\multv{V}_S}
\newcommand{\strans}{$\overline{\phantom{M}}^S$}
\newcommand{\sleq}{\rightsquigarrow_s}

\begin{definition}{Translation of {\lamadio} in {\lthis}}\\
  Let $M$ be a term and $S$ be a store of \lamadio.
  S can be written as
  \[
    S = r_1 \Leftarrow V^1_1 \parallel \ldots \parallel r_1
    \Leftarrow V^1_{k_1} \parallel \ldots \parallel r_n \Leftarrow
    V^n_1 \parallel \ldots \parallel r_n \Leftarrow
    V^n_{k_n}
  \]
  Let $\multv{V}_S : r_i \mapsto [V^i_1, \ldots, V^i_{k_i}]$.
  We define the translation of $M$ under $S$ by
  \begin{itemize}
    \item If $M = V$ or $M = \set{s}{V}$ then $\overline{M}^S = M$
    \item If $M = N\ N'$ then $\overline{M}^S = \refssl{}{\vstore}{\overline{N}^S\
      \overline{N'}^S}$
    \item If $M = \get{s}$ then $\overline{M}^S =
      \refssd{}{\vstore}{\get{s}}$
    \item If $M = N \parallel N'$ then $\overline{M}^S = \overline{N}^S
      \parallel \overline{N'}^S$
  \end{itemize}

  In fact, $\overline{M}^S$ is the normal form reached from
  $\refsrawsd{r}{\vstore}{M}$ using only downward structural rules.
  For any program $P = M \parallel S$, we define $\overline{P} =
  \overline{M}^S$.
  If $P = M$ then $\overline{P} = \overline{M} = M$.
\end{definition}

Our reduction have the drawback of not reducing under abstractions, such that
some of the variable substitution propagation will be delayed until the
corresponding lambda will be applied (or forever). To cope with this subtlety,
we introduce a relation on terms of {\lthis} that expresses that a
term $M$ is related to $N$ if $M$ is the same as $N$ up to some pending
substitutions hidden under lambdas, and such that if we could reduce these
substitutions freely $M$ would actually reduce to $N$.

  \begin{definition}{Substitution relation}\label{mdef-substitution-equivalence}\\
    We define $\sleq$ as :
    \begin{itemize}
      \item $x \sleq x$, $\get r \sleq \get r$, $\ast \sleq \ast$
      \item $M \parallel N \sleq M' \parallel N'$ iff $M \sleq M'$ and $N \sleq N'$
      \item $\lambda x.M \sleq \lambda x.M'$ if $M =
        \varss{\sigma_1}{\varss{\ldots}{\varss{\sigma_n}{N}}}$ such that
        $N\{\sigma_1,\ldots,\sigma_n\} = M$.
      \item $\refssd{}{V}{M} \sleq \refssd{}{V'}{M'}$ or $\refssu{}{V}{M}
        \sleq \refssu{}{V'}{M'}$ iff $M \sleq M'$ and $\multv{V} \sleq
        \multv{V'}$
      \item $\refssl{}{V}{M\ N} \sleq \refssl{}{V'}{M'\ N'}$ iff $M \sleq M'$,
        $N \sleq N'$ and $\multv{V} \sleq \multv{V'}$
      \item $\varss{\sigma}{M} \sleq \varss{\sigma'}{M'}$ iff $M \sleq M'$ and
        $\sigma \sleq \sigma'$
    \end{itemize}
    Where we extended point-wise the definition of $\sleq$ to functions and
    multisets.
  \end{definition}

\begin{theorem}{Simulation}\\
  Let $P,Q$ be {\lamadio} programs such that $P \rightarrow^\ast Q$. Then
  $$\overline{P} \rednd^\ast M \sleq \overline{Q}$$
\end{theorem}

\end{document}